\documentclass[journal,12pt,onecolumn]{IEEEtran}
\usepackage{caption}
\usepackage{amssymb, amsmath, amsthm}
\usepackage{amsfonts, bbm}
\usepackage{graphicx}
\usepackage{epstopdf}
\usepackage{times}
\usepackage{float} 
\usepackage{array}
\usepackage{arydshln}
\usepackage{bibentry}
\usepackage{booktabs}

\usepackage{cite}
\usepackage{subfigure}
\usepackage{enumitem}
\usepackage{color}
\usepackage{setspace}
\usepackage{bigstrut}
\usepackage{multirow}

\allowdisplaybreaks
\theoremstyle{plain}

\theoremstyle{plain}
\newtheorem{corollary}{Corollary}
\theoremstyle{plain}
\newtheorem{lemmacounter}{Theorem}
\newtheorem{lemma}[lemmacounter]{Lemma}
\theoremstyle{plain}
\newtheorem{defcounter}{Theorem}
\newtheorem{definition}[defcounter]{Definition}
\theoremstyle{plain}

\newtheorem{special case}[defcounter]{Special Case}
 
\newcommand*{\Scale}[2][4]{\scalebox{#1}{$#2$}}%

\newcommand{\figref}[1]{Fig.~\ref{#1}}
\begin{document}

\title{Modeling and Analysis of Uplink Non-Orthogonal Multiple Access (NOMA) in Large-Scale Cellular Networks Using  Poisson Cluster Processes}
\author{Hina Tabassum, Ekram Hossain, and Md. Jahangir Hossain\thanks{H. Tabassum and E. Hossain are with the Department of Electrical and Computer Engineering at the University of Manitoba, Canada.
 Md. J. Hossain is with the University of British Columbia (Okanagan Campus), Canada.
 }}
\maketitle

\begin{abstract}
Non-orthogonal multiple access (NOMA) serves multiple users by superposing their  distinct message signals. The desired message signal is decoded at the receiver by applying  successive interference cancellation (SIC). 
Using the theory of Poisson cluster process (PCP), this paper provides a framework to analyze multi-cell uplink NOMA systems. Specifically,
we characterize the rate coverage probability of a NOMA user who is at rank $m$ (in terms of the distance from its serving  BS) among all users in a cell and the mean rate coverage probability of all users in a cell. Since the signal-to-interference-plus-noise ratio  (SINR) of $m$-th user relies on efficient SIC, we consider three scenarios, i.e., perfect SIC (in which the signals of $m-1$ interferers who are stronger than $m$-th user are decoded successfully), imperfect SIC (in which the signals of of $m-1$ interferers who are stronger than $m$-th user may or may not be decoded successfully), and imperfect worst case SIC (in which the decoding of the signal of $m$-th user  is always unsuccessful whenever the decoding of its relative $m-1$ stronger users is unsuccessful).  The worst case SIC assumption provides remarkable simplifications in the mathematical analysis and is found to be highly accurate for scenarios of practical interest.
To analyze the rate coverage expressions, we first characterize the Laplace transforms of the intra-cluster interferences in closed-form considering both perfect and imperfect SIC scenarios. In the sequel, we characterize the  distribution of the distance of a user at rank $m$ which is shown to be the generalized Beta distribution of first kind and the conditional distribution of the distance of the intra-cluster interferers which is different for both perfect and imperfect SIC scenarios.
The Laplace transform of the inter-cluster interference is then characterized by exploiting distance distributions from geometric probability. The derived expressions are customized to capture the performance of a user at rank $m$ in an equivalent orthogonal multiple access (OMA) system. Finally, numerical results are presented to validate the derived expressions. It is shown that the average rate coverage of a NOMA cluster outperforms its counterpart OMA cluster  with higher number of users per cell and higher target rate requirements. A comparison of Poisson Point Process (PPP)-based and PCP-based modeling is conducted which shows that the PPP-based modeling provides optimistic results for the NOMA systems.
\end{abstract}


\section{Introduction}

Until very recently, the state-of-the-art wireless communications systems have been utilizing a variety of orthogonal multiple access (OMA) technologies, in which the resources are allocated orthogonally to multiple users.  These techniques include frequency-division
multiple access (FDMA), time-division multiple access
(TDMA), code-division multiple access (CDMA), and
orthogonal frequency-division multiple access (OFDMA).  Since OMA maintains the orthogonality among users in a cell, the intra-cell interference (i.e., inter-user interference within a cell) does not exist. As a result, the  information signals of users can be retrieved at a low complexity. Nonetheless, the number of served users is limited by the number of orthogonal resources. 

Conversely, NOMA serves  multiple users  simultaneously  using the same spectrum resources (i.e., radio channels), however, at the cost of increased intra-cell interferences.
To mitigate the intra-cell interferences, NOMA exploits Successive Interference Cancellation (SIC) at the receivers~\cite{saito2013}.  NOMA supports low transmission latency and signaling cost compared to  conventional OMA  where each user is obliged to send  a channel scheduling request to its serving base station (BS).   With these attractive features, NOMA can be a potential access technology for 5G networks.  Nevertheless, the conclusions about the performance of NOMA are largely unknown in multi-cell network scenarios. For instance, in uplink NOMA, a large number of transmitting users
in the neighboring co-channel BSs can result in  high interference at the BS of interest. Consequently, the uplink multi-cell interference in NOMA is directly proportional to the number  of  transmitting users per neighboring co-channel BS and is more severe compared to  the inter-cell interference in OMA.

\subsection{Background Work}
The concept of NOMA  was initially proposed in \cite{saito2013} for downlink transmissions. Various practical aspects, such as multi-user scheduling, impact of error propagation in SIC, overall system overhead, and user mobility were discussed. System level simulations were conducted in \cite{ben2015} to highlight the benefits of two-user NOMA over  OMA, in terms of overall system throughput as well as individual user's throughput.

The approximate expressions for the ergodic sum-rate and outage probability of a user in a given downlink NOMA cluster were derived in \cite{ding2014}. 
Later, in \cite{ding2015}, the throughput gains of the two-user cooperative NOMA (in which the strong channel user relays the information of weak channel user) were investigated. 
The idea of   cooperative  NOMA was then applied to wireless-powered systems  in \cite{liu}. Based on users' distances, grouping of users was  performed first. Then, three  user  selection  schemes  were investigated, i.e., (i)~pairing of the nearest users from each group, (ii) pairing of the nearest user from one group and the farthest  user from another group, and (iii) arbitrary user pairing. The direct link was used to transfer energy from the BS. A cooperative data link was established for the lower channel gain user via the higher channel gain users. Closed-form approximate expressions for  the  outage  probability  and   throughput of a two-user NOMA cluster were derived. 
In \cite{dingearly}, the user pairing was investigated  considering  fixed NOMA (F-NOMA) and cognitive radio inspired (CR-NOMA). In F-NOMA, any two users could make a NOMA pair  based on their channel gains.  While in CR-NOMA, a weak channel user opportunistically gets paired with the strong channel user provided  that the interference caused by the strong user will not harm the rate requirement of the weak channel user. It was observed that  CR-NOMA  pairs the  strongest user with the second strongest user, whereas F-NOMA pairs a strongest user with the weakest user in the system.

A general concept of uplink NOMA was discussed in \cite{zhang2016}. An uplink power back-off policy was proposed to distinguish users in a NOMA cluster with nearly similar signal strengths (given that traditional uplink power control is applied). Closed-form analysis was performed for ergodic sum-rate and outage probability of  a two-user NOMA cluster.  Further, the problem of user scheduling, subcarrier allocation, and power control in uplink NOMA has been investigated by various researchers in~\cite{uplink1,uplink2} with perfect SIC at the BS. A game theoretic algorithm for uplink power control has been designed in \cite{uplink3} considering a two-cell NOMA system where inter-cell interference is assumed to be Gaussian distributed.

\subsection{Motivations and Contributions}
To date, most of the research investigations consider the throughput analysis of NOMA  for single-cell downlink systems with perfect SIC   at the receivers. The derived expressions generally leverage on high signal-to-noise ratio (SNR) and asymptotic assumptions as well as the  application of Gaussian-Chebyshev quadrature (GCQ) technique which approximates all integrals into finite sums. 
Unfortunately, the conclusions about the performance gains of NOMA (compared to OMA) in single-cell/single-cluster scenarios cannot be applied directly to  multi-cell/multi-cluster scenarios. The reason is the inter-cell\footnote{The term inter/intra-cell and inter/intra-cluster interference will be used interchangeably throughout the paper.} interferences in NOMA can be quite severe as well as distinct from OMA, especially  in the uplink scenarios~\cite{nomamag}. Note that, in uplink NOMA, the inter-cell interference incurred at the BS of interest is directly proportional to the number of transmitting users in  neighboring co-channel BSs. This is different from an equivalent uplink OMA system where only one user transmits at a time per neighboring co-channel BS. 

Further, compared to downlink NOMA, we note that the performance analysis of uplink NOMA  is particularly more challenging due to the mathematical structure of the intra-cell interferences. 
In the uplink NOMA, the BS receives transmissions from all users
simultaneously. As such, the intra-cell interference to a user
is a function of the channel statistics of other users within
the cell. On the other hand, in downlink NOMA, the intra-cell interference to a user is a function of its own channel statistics~\cite{ding2014,ding2015,dingearly}. To this end, the contributions of this paper are outlined as follows:
\begin{itemize}
\item Using the theory of order statistics and Poisson Cluster Process\footnote{Neyman-Scott PCP, such as Modified Thomas Cluster process and Matern Cluster Process, are recently exploited in a set of research studies for  performance evaluation of ad-hoc clustered networks~\cite{martin}, D2D systems~\cite{afshang1,afshang2}, and downlink multi-tier cellular networks~\cite{hasna,afshang0}.} (PCP), we develop a framework to analyze 
 the rate coverage probability of a user who is at rank $m$ (in terms of the distance from its serving BS) among all users in a cell and the mean rate coverage probability of all users in a cell, considering a  multi-cell uplink NOMA system.  Unlike typical stochastic geometry frameworks where the user locations are uniform over the 2-D plane and are independent of the BS locations~\cite{hesham1,harpreet}, we exploit Matern Cluster Process (MCP) to accurately model the  proximity of multiple users around a BS\footnote{{The Poisson Point Process (PPP)-based modeling of BSs and user devices focuses on the link between the serving BS and a typical user. Since the typical user can be located anywhere in the cell, results are averaged over all spatial positions inside the cell. 
Such an approach provides higher analytical flexibility. 
However, in practice, users are more likely clustered around a BS and are distinct due to their channel conditions. In this regard, PCP  have been shown empirically to be a more accurate cellular network modeling technique~\cite{iccc} that allows location-specific performance modeling of users.}}.

\item  NOMA systems rely on efficient SIC and the interference  of a user at rank $m$ needs to be adapted according to the level of  SIC. We consider three SIC scenarios that include perfect SIC (in which the signals of $m-1$ interferers who are stronger than $m$-th user are decoded successfully), imperfect SIC (in which the signals of $m-1$ interferers who are stronger than $m$-th user may or may not be decoded successfully), and imperfect worst case SIC (in which the decoding of the signal of $m$-th user  is always unsuccessful whenever the decoding of its relative $m-1$ stronger users is unsuccessful).  The worst case SIC assumption provides remarkable simplifications in the mathematical analysis and is found to be highly accurate for scenarios of practical interest.

\item  The interference power at the BS of a given cell/cluster is composed of intra-cluster and inter-cluster interferences. We derive the Laplace transform of the intra-cluster interference in closed-form considering various SIC scenarios. In the sequel, we also derive the distribution for the distance of a user at rank $m$ which is shown to be the generalized Beta distribution of first kind and the conditional distribution of the distance of the intra-cluster interferers which is different for both perfect and imperfect SIC scenarios. The Laplace transform of the inter-cluster interference is then characterized using distance distributions from geometric probability. A less-complex bound is then exploited to model the Laplace transform of the inter-cluster interference.

\item The derived rate coverage expressions are  customized to evaluate the performance of a user at rank $m$ in an equivalent OMA system in closed-form. Numerical results are presented to validate the derived expressions. 
Our results indicate that the performance benefit of OMA diminishes quickly with  the increase in number of users per cluster and higher user rate requirements. A comparative performance analysis of PPP-based and PCP-based modeling is conducted using simulations. 
It is shown that  PPP-based modeling generally provides optimistic results for the NOMA systems due to the homogeneous distribution of users regardless of the BS locations which reduces the impact of intra-cluster interference.
\end{itemize}
  
The rest of the paper is structured as follows. Section~II discusses the working principle of uplink NOMA along with cellular network model, channel model, and interference model. In Section~III, we describe the fundamental differences between conventional SIC and SIC for NOMA. Considering  perfect SIC, imperfect SIC, and imperfect worst case SIC, we model the interferences and define the desired performance metrics. In Section~IV, we derive relevant distance distributions required for the characterization of the Laplace transforms of the interferences. In Section~V, we derive the rate coverage  expressions for both NOMA and OMA systems. Finally, Section~VI discusses numerical and simulation results followed by the concluding remarks in Section~VII.

\section{System Model and Assumptions}

\begin{figure}[t]
\begin{center}
\includegraphics[width = 3.5in]{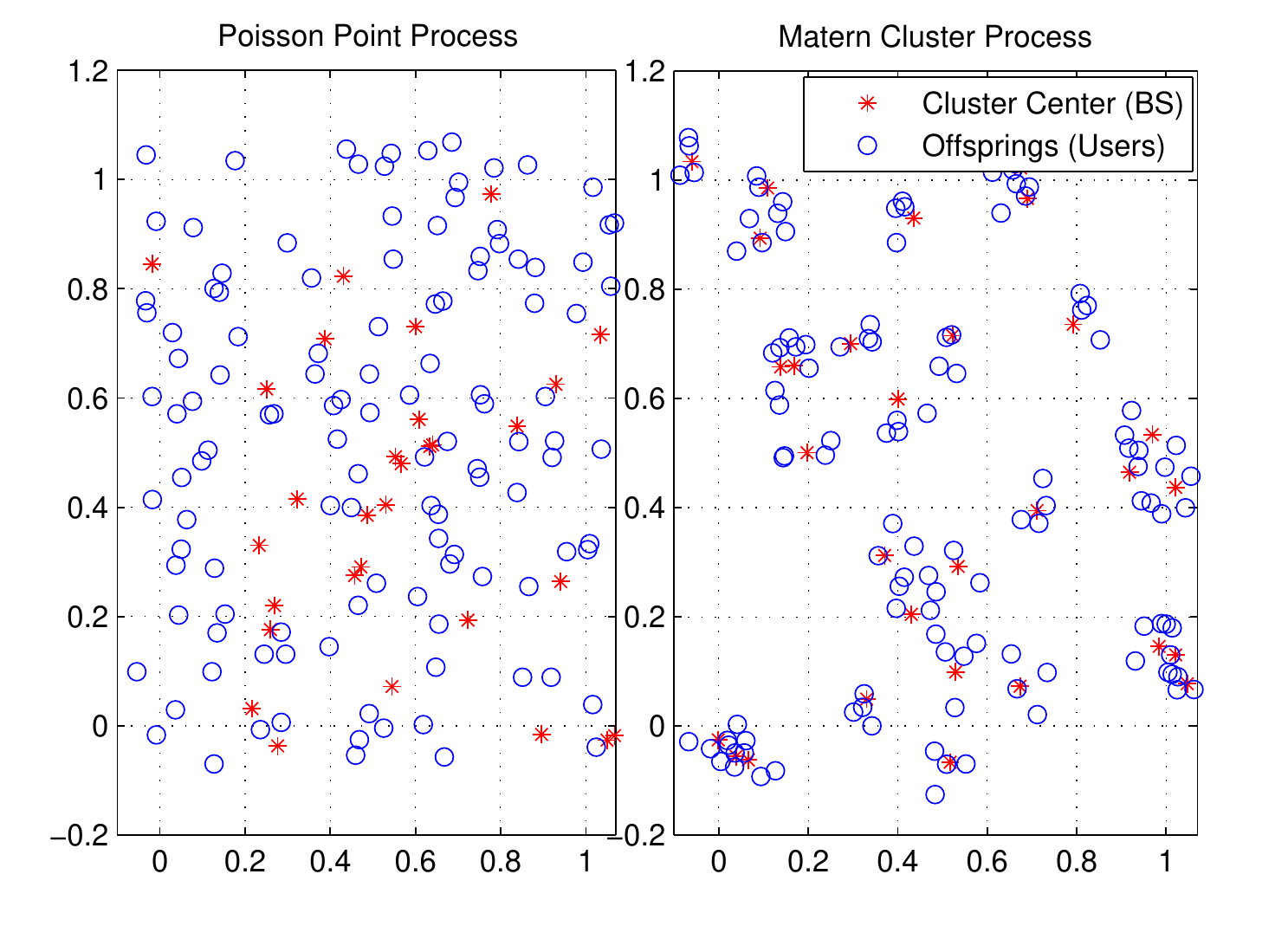}
\caption[c] { Graphical illustration of the single-tier cellular network with (a)~conventional PPP model and (b)~Matern Cluster Process (MCP) model in which the number of users are Poisson distributed and their locations are uniformly distributed within a disc of radius $R$ per cluster center.}
\label{faircomp2}
\end{center}
\end{figure}

\subsection{Spatial Cellular Network Model}

We consider  a single-tier cellular network composed of macrocell base stations (BSs) surrounded by user devices. The locations of transmitting user devices  are modeled as a stationary  and isotropic  PCP.
\begin{definition}[Poisson Cluster Process (PCP)~\cite{martin}]
A PCP  results from applying homogeneous independent
clustering to a stationary Poisson Point Process (PPP). In particular, the parent points form a  stationary PPP $\phi=\{x_1,x_2,\cdots\}$  of density $\lambda$. The  off-springs $N^{x_i}$ generated for a given parent $x_i$ are a family of independent and identically distributed (i.i.d.) finite point sets with distribution independent of the parent process. The complete PCP can thus be represented as
$
\Xi=\underset{x \in \Phi}{\cup} N^{x}.
$
Note that the parent points themselves will not be included in the PCP. The parents and off-springs are referred to as the cluster centers  and the cluster members, respectively. 
\end{definition}
\begin{definition}[Neyman-Scott PCP]
If the number of points per cluster follow a Poisson distribution with mean intensity $\bar{c}$, such a PCP is referred to as a Neyman-Scott Process. 
\end{definition}
MCP is a special case of the Neyman-Scott PCP  where the
cluster centers (BSs) are modeled by a parent homogeneous PPP $\Phi_m = \{x_0, x_1,x_2,\cdots\}$  in the Euclidean plane with density $\lambda_m$. 
Each parent point $x_i \in \Phi_m$ forms the center of a cluster around which  daughter points (user devices) are uniformly distributed in a circle of radius $R$ as shown in Fig.~1. For ease of exposition, we consider $\bar{c}$ number of users per cluster.
Each daughter point located at $y$ with respect to  its cluster center $x_i$ has a density function as:
\begin{equation}\label{fy}
f(y)=\frac{1}{\pi R^2}, \:\:\:||y|| \leq R,
\end{equation}
where $||y||=r$ is the  distance of any arbitrary daughter point (user device) relative to its cluster center (serving BS) and its density function can be given as follows:
\begin{equation}\label{fr}
f(r)=\frac{2 r}{ R^2}, \:\:\: r\leq R.
\end{equation}
The resulting MCP is  a stationary and isotropic point process of density $\bar{c} \lambda_m$ and can be defined as $
\Xi=\underset{x_i \in \Phi_m}{\cup} \mathcal{N}^{x_i}.
$
Without loss of generality, we perform analysis for a user at rank $m$ located in a randomly chosen cluster which is referred as representative cluster located at $x_0$ throughout the paper.

\subsection{Working Principle of Uplink NOMA}
In an uplink NOMA cluster,  each user transmits its individual signal $\hat x_i$ with a transmit power $p_i$ such that the received signal at the BS can be defined as $\hat y=\sum_{i=1}^{\bar{c}} \sqrt{p_i} h_i \hat x_i$. 
Note that, to apply SIC  and decode signals at the BS, it is crucial to maintain the distinctness of various signals superposed within $\hat y$. 
Since the channels of different users are different in the uplink, each message signal experiences distinct channel gain\footnote{The conventional uplink transmit power control (typically intended to equalize the received signal powers of users) may remove the channel distinctness and thus may not be feasible for uplink NOMA transmissions.}. 
As a result, the received signal power corresponding to the strongest channel user is likely the strongest at the BS. Therefore, this signal is decoded first at the BS  and experiences interference from all  users in the cluster with relatively weaker channels. That is,  the transmission of the highest channel gain user experiences  interference from all users within its cluster, whereas the transmission of  the lowest channel gain user receives zero interference from the users in its cluster.

\subsection{Channel and Interference Model}
NOMA allows multiple users in a cluster to share the same 
resources by superposing their distinct message signals. 
All users are served on the same channel and time slot. All users/BSs are equipped with a single antenna.
Within the representative cluster, any arbitrary user device located at $y \in \mathcal{N}^{x_0}$ with respect to its serving BS located at $x_0 \in \Phi_m$ transmits its individual signal with power $P_u$ such that the superposed  NOMA signal at the representative BS can be defined as follows:
\begin{equation}
S_0=\sum_{y \in \mathcal{N}^{x_0}} S_y=\sum_{y \in \mathcal{N}^{x_0}} {P_u} h_{y_{x_0}} ||y||^{-\alpha},
\end{equation} 
where 
$\alpha>2$ is the path-loss exponent and $h_{y_{x_0}}$ is the exponential random variable which models Rayleigh fading associated
with the channel between the node located at $y$ and the BS located at $x_0$. All fading coefficients are i.i.d. and the additive noise is complex Gaussian distributed with mean zero and variance $N_0/2$ per dimension.

Consequently, the uplink NOMA transmission of a user located at $y_0$ within the representative cluster is vulnerable to two  kind of interferences, i.e., 
\begin{itemize}
\item {\it Intra-cluster interference:}  is the interference received at the  representative BS from all user devices located within the representative cluster (except the user located at $y_0$). However, after performing SIC, some of these interferences can be removed
(details of SIC will follow in the next section). 
\item {\it Inter-cluster interference:} is the interference received at the representative BS from all user devices located outside the representative cluster. 
\end{itemize}
Provided that the center of the representative cluster is located at $x_0 \in \Phi_m$, the intra-cluster interference experienced by the transmission of user  located at $y_0$ with respect to its cluster center $x_0$ can   be modeled as follows:
\begin{equation}
I_{\mathrm{intra}}=\sum_{y \in \mathcal{N}^{x_0} \backslash y_0}  P_u h_{y_{x_0}} ||y||^{-\alpha}.
\end{equation}
Similarly, the inter-cluster interference at the representative BS from the user devices located outside the representative cluster center can be modeled as follows:
\begin{equation}\label{inter}
I_{\mathrm{inter}}=\sum_{x \in \Phi_m \backslash x_0} \sum_{y \in \mathcal{N}^{x}}  P_u h_{y_{x}} ||x+ y||^{-\alpha},
\end{equation}
and $I_{\mathrm{agg}}=I_{\mathrm{intra}}+I_{\mathrm{inter}}$.
However, since NOMA systems rely on efficient SIC, the intra-cluster interference model needs to be adapted according to the level of SIC cancellations. This is elaborated in detail in the next section.

\section{Successive Interference Cancellation (SIC) and Performance Metrics}
In this section, we will discuss the fundamental differences between conventional SIC and SIC for NOMA. We then describe the signal-to-interference-plus-noise ratio (SINR) modeling with perfect and imperfect SIC in uplink NOMA. Considering three possible SIC scenarios, i.e., perfect SIC, imperfect SIC, and imperfect worst case SIC, we finally define the performance metrics.

SIC is among one of the  best known interference cancellation methods as (i)~the SIC receiver is architecturally similar to traditional non-SIC receivers in terms of hardware complexity and cost, (ii)~it uses the traditional decoder to decode the composite signal at different stages and neither complicated decoders nor
multiple antennas are required, and (iii)~it can achieve the Shannon capacity for both the broadcast and multiple access networks~\cite{SIC1,SIC2}.
Typically, SIC is used to regenerate the interfering signals and
subsequently cancel them from the received composite signal to improve the SINR of the desired signal. That is, the SIC receiver first decodes the strongest signal by treating other signals as
noise. Then it regenerates the analog signal from the decoded
signal and cancels it from the received composite signal. The remaining signal is thus free from the  strongest interfering signal. Then, the SIC receiver
proceeds to decode, regenerate, and cancel the second strongest
interfering signal from the remaining signal and so on,
until the desired signal can be decoded.  

\subsection{SIC and SIC Error Propagation in Uplink NOMA}
In uplink NOMA, we apply the same SIC principle at the BS, i.e., the SIC receiver first decodes the strongest signal by treating other signals as noise and so on.  However, the difference is that the intra-cluster interfering signals are also the desired signals; therefore, it is not possible to provide the benefits of SIC (enhance the SINR) unequivocally for all users. 
That is, within a cluster,  the user with strongest signal experiences  interference from all users and  the user with the weakest signal enjoys zero intra-cluster interference. {\em Evidently, the intra-cluster interference statistics will vary for all user devices within the representative NOMA cluster.}

Since the decoding of the strongest signal is performed first at the BS, its success/failure has a significant impact on the decoding of other users' signals.  Specifically, depending on the decoding result of the strongest signal, the interference used for the decoding of the second strongest signal  differs, which makes the link-to-system mapping difficult. If the strongest signal is decoded correctly, its replica signal can be  subtracted successfully from the superposed signal at the BS.  Otherwise, the second strongest user will experience the interference from the  strongest user as well as other users in the cluster. This phenomenon is referred to as {SIC error propagation}.

\subsection{Modeling of Intra-cluster Interference with SIC}
To model the intra-cluster interference with SIC, first the  BS needs to  rank the received powers of various users  as
$\{S_{(1)},S_{(2)}, \cdots, S_{(m)}, \cdots, S_{(\bar{c})}\}$ such that $S_{(m)} \geq  S_{(\bar c)}$ with $m < \bar c$. 
However, note that the impact of path-loss factor is more stable and  dominant compared to the instantaneous multi-path channel fading effects.  Therefore, the order statistics of the distance outweigh the fading effects, which vary on a much shorter time scale.
As such, the ranking of users in  terms of their distances from the serving BS is generally considered as a reasonable approximate of their respective ranked received signal powers~\cite{approx1}. This approximation provides great flexibility  for the analytical purposes. Also, the SIC based on long-term channel states is more practically feasible since it requires less overheads for channel estimation. Note that the exact performance analysis of the user with $m^{\mathrm{th}}$ strongest signal is unwieldy to solve since the distribution of $S_{(j)}$ is the ranked distribution of a composite uniform and exponential random variable and  the joint distribution of several composite ordered random variables is required.

{\bf Approximation:} As mentioned above, the impact of path-loss factor is more dominant compared to the channel fading effects. Hence, for tractability reasons, we assume that ordering of the received signal powers   can be approximately achieved by ordering the distances of the users as $r_{(1)} \leq r_{(2)}, \cdots,\leq r_{(m)}, \cdots,\leq r_{(\bar c)}$ such that $r_{(m)} \leq  r_{(\bar c)}$ with $m < \bar c$ . That is, when the  $j$-th strongest signal  is decoded and subtracted from the composite signal, this means that the remaining interferers are located farther than the $j$-th rank user whose distance is $r_{(j)}$ from the representative BS.

The intra-cluster interference and in turn the SINR of $m$-th rank user can thus be modeled for perfect and imperfect SIC scenarios, respectively, as shown below.
\subsubsection{Perfect SIC} 
In this case, a given user at rank $m$  receives interferences from all users with relatively weaker channel gains (or users with farther distances as per the approximation) and the BS perfectly decodes/cancels the $m-1$ strong interferences. The intra-cluster interference experienced by any user  at $m$-{th} rank in a cluster can thus be modeled, after {\em perfectly} canceling $m-1$ strong interferences, as follows:
\begin{equation}\label{perfectI}
I^{m}_{\mathrm{intra}}= 
\sum_{j=m+1}^{\bar{c}}  S_{(j)}
\approx
\sum_{\substack{j=m+1\\y_j\in \mathcal{N}^{x_0}}}^{\bar{c}}  P_u h_{{y_j}_{x_0}} ||y_{(j)}||^{-\alpha},
\end{equation}
where $||y_{(j)}||=r_{(j)}$. Note that the ranking is applied only at the distances $||y_{(j)}||^{-\alpha}$  as per the approximation.
The SINR experienced by any user  at $m^{\mathrm{th}}$ rank in the representative cluster can therefore be defined as follows:
\begin{equation}\label{perfectsic}
\mathrm{SINR}_{m}=\frac{S_{(m)}}
{I^{m}_{\mathrm{agg}}}
=\frac{S_{(m)}}
{I^{m}_{\mathrm{intra}}+I_{\mathrm{inter}}+N_0}.
\end{equation}

\subsubsection{Imperfect SIC and Detection Probability}
The signals from $m-1$  interferers (who are located closer to the BS than the rank $m$ user) may or may not be decoded perfectly; therefore, SIC may or may not be performed in a perfect fashion. 
In such a case, we first define the probabilities for the successful detection of the signals of  the ranked users (ranked in terms of their distances), respectively, as follows:
\begin{align*}
&p_{(1)}=\mathbb{P}\left(\frac{S_{(1)}}{I^{1}_{\mathrm{agg}}} \geq \theta\right),
\nonumber\\
&
p_{(2)}=
p_{(1)}
\mathbb{P}\left(\frac{S_{(2)}}{I^{2}_{\mathrm{agg}}} \geq \theta\right)
+
\bar{p}_{(1)}
\mathbb{P}\left(\frac{S_{(2)}}{I^2_{\mathrm{agg}}+S_{(1)}} \geq \theta\right),
\nonumber\\
&
p_{(3)}=
p_{(1)} p_{(2)}
\mathbb{P}\left(\frac{S_{(3)}}{I^{3}_{\mathrm{agg}}} \geq \theta \right)
+
\bar{p}_{(1)}p_{(2)}
\mathbb{P}\left(\frac{S_{(3)}}{I^{3}_{\mathrm{agg}}+S_{(1)}} \geq \theta \right)
\nonumber\\&+
\Scale[1]{\bar{p}_{(2)} p_{(1)}
\mathbb{P}\left(\frac{S_{(3)}}{I^{3}_{\mathrm{agg}}+S_{(2)}} \geq \theta \right)
+
\bar{p}_{(1)}\bar{p}_{(2)}
\mathbb{P}\left(\frac{S_{(2)}}{I^{3}_{\mathrm{agg}}+\sum_{j=1}^ 2 S_{(j)}} \geq \theta \right)},
\end{align*}
where $p_{(j)}$ and $\bar{p}_{(j)}, \forall j=1,2,3,\cdots$,  represent the probability of successful and unsuccessful detection of $j$-th ranked user's signal, respectively. It can be seen that the BS attempts to decode the closest user without
any interference cancellation. If the decoding is unsuccessful, the interference from this user remains intact. Subsequently, we can generalize the detection probability of a user at $m$-{{th}} rank as follows:
\begin{align}\label{pm}
p_{(m)}=&\sum_{b\in \mathcal{B}} 
A
\mathbb{P}\left(\frac{S_{(m)}}{I^{m}_{\mathrm{agg}}+\sum^{m-1}_{j=1}{(1-b(j))S_{(j)}}} \geq \theta\right),
\end{align}
where $A=\left(\prod^{m-1}_{j=1} (p_{(j)})^{b(j)} (\bar p_{(j)})^{1-b(j)}\right)$, $\theta$ is the signal detection threshold, and $\mathcal{B}$
denotes the set of $2^{m-1}$ combinations in which each combination $\mathbf{b}$ has $m-1$ bits. The successful detection is represented by a binary digit $b(j)=1$ whereas the detection failure is given by $b(j)=0$.
The intra-cluster interference experienced by a user at rank $m$ thus depends on whether the detections  for $m-1$ closer users were successful or not. As such,  conditioned on a given combination $\mathbf{b}$, the SINR experienced by a user at rank $m$ can  be modeled as:
\begin{align}\label{imperfectI}
\mathrm{SINR}_{m,\mathbf{b}}=\frac{S_{(m)}}{I^{m,\mathbf{b}}_{\mathrm{intra}}+I_{\mathrm{inter}}+N_0},
\end{align}
where
$$I^{m,\mathbf{b}}_{\mathrm{intra}}=I^{m}_{\mathrm{intra}}+I^{m}_{\mathrm{add}}=I^{m}_{\mathrm{intra}}+\sum^{m-1}_{j=1}{(1-b(j))S_{(j)}}.$$
Note that, even after successful detection, a given user can still experience  rate outage (i.e., the achievable rate may remain below the target rate requirement).

\subsection{Performance Metrics}
We analyze the performance gains of uplink NOMA  considering a system where each user has a target data rate requirement. For this case, the rate coverage probability of a user at rank $m$ and the mean rate coverage probability of a cluster are  relevant performance metrics. These metrics are defined for different SIC scenarios in the following.
\subsubsection{Rate Coverage}
Rate coverage probability  is the probability that a given user's achievable rate remains above the target data rate. Mathematically, the rate coverage probability of a user at $m$-th rank can be defined as  $\mathbb{P}(\mathrm{log}_2(1+\mathrm{SINR}_m) \geq R_m)$, where $R_m$ is the target data rate requirement of the $m$-th ranked user.
Now, we define the rate coverage probability of $m$-th user in the following specific cases:
\begin{itemize}
\item {\em Rate Coverage with Perfect SIC:} Using the definition of $\mathrm{SINR}_{m}$ from \eqref{perfectsic}, the  rate coverage probability of a user at rank $m$ in the representative cluster can be defined as: 
\begin{align}\label{best}
\mathcal{C}^{(\mathrm{P})}_m=\mathbb{P}(\mathrm{SINR}_{m} \geq \gamma_m)
=\mathbb{P}\left(\frac{S_{(m)}}{I^{m}_{\mathrm{agg}}} \geq \gamma_m\right),
\end{align}
where $\gamma_m=2^{R_m}-1$ is the desired SINR corresponding to the rate requirement $R_m$ of the user at rank $m$. 
\item {\em Rate Coverage with Imperfect SIC:} In this case, we consider the probability of decoding/canceling interferences from $m-1$ closer interferers as less than one.  As such, the rate coverage probability of a user at rank $m$ needs to consider all possible combinations $\mathbf{b} \in \mathcal{B}$ and thus can be defined using \eqref{imperfectI} as follows:
\begin{align}
\mathcal{C}^{(\mathrm{I})}_m
&=
\sum_{b\in \mathcal{B}} A
\mathbb{P}\left({\mathrm{SINR}_{m,\mathbf{b}}} \geq  \gamma_m\right).
\end{align}

\item {\em Rate Coverage with Imperfect SIC - Worst Case:}
The worst-case model assumes that the decoding of any user at rank $m$ is always unsuccessful whenever the decoding of his relative $m-1$ closer users is unsuccessful. Such a worst-case
model is simple and allows evaluating the impact of SIC error propagation on the NOMA performance without invoking complicated NOMA specific link-to-system mapping. The worst-case detection probability of a user at rank $m$ can therefore be given as:
\begin{equation}\label{pworst}
{p}^{\mathrm{worst}}_{(m)}= \prod_{i=1}^{m-1} \mathbb{P}(\mathrm{SINR}_{i} \geq \theta),
\end{equation}
where $\mathrm{SINR}_{i}$ can be given using \eqref{perfectsic}.
The rate coverage probability of a user at rank $m$ can thus be given as:
\begin{equation}
\mathcal{C}^{\mathrm{worst}}_{m}={p}^{\mathrm{worst}}_{(m)}\mathcal{C}^{(\mathrm{P})}_m.
\end{equation}
\end{itemize}
As a by-product, we can evaluate the rate coverage of a user at rank $m$ in an equivalent TDMA-based OMA system where $\bar c$ users in a cluster are served in orthogonal time slots. The rate of a user at rank $m$ in OMA system can  be defined as $R^{(\mathrm{oma})}_m= \frac{1}{\bar{c}} \mathrm{log}_2 (1+\mathrm{SINR}^{\mathrm{oma}}_m)$. For fair comparison with NOMA, we need to include the scaling factor of $\frac{1}{\bar{c}} $ which shows the portion of available resources to a given user in OMA system. As such, if the target rate requirement of the user is $R_{\mathrm{th}}$, the desired SINR threshold  $\gamma^{(\mathrm{oma})}_m$ can be defined as $2^{R_{\mathrm{th}} \bar{c}}-1$. 
Subsequently, substituting zero intra-cluster interference, $\bar{c}=1$ for inter-cell interfering clusters, and replacing $\gamma_m=\gamma^{(\mathrm{oma})}_m$ in the rate coverage expressions of NOMA, we can determine the rate coverage of a user at rank $m$ in an OMA system.

Further, we can also calculate the average rate of a user at rank $m$. For  perfect SIC, imperfect SIC, and worst case SIC scenarios,  average rates for a user at rank $m$ can be computed, respectively, as follows: 
\begin{align*}
&
\mathcal{R}^{(\mathrm{P})}_m
=\mathbb{E}[\mathrm{ln}\left(1+\mathrm{SINR}_m\right)]=\int_0^\infty \mathbb{P}(\mathrm{SINR}_m>e^t-1)dt,
\\
&
{\mathcal{R}^{(\mathrm{I})}_m
=\sum_{b\in \mathcal{B}} A
\int_0^\infty \mathbb{P}(\mathrm{SINR}_{m,\mathbf{b}}>e^t-1)dt},
\\
&
\mathcal{R}^{\mathrm{worst}}_m
={p}^{\mathrm{worst}}_{(m)}\mathcal{R}^{(\mathrm{P})}_m.
\end{align*}
Substituting $e^{t}-1$  in place of $\gamma_m$, we can integrate all coverage probability expressions over $t$ to evaluate the average rate of a user at rank $m$ numerically.
In addition, substituting $m=1$ and $m=\bar{c}$, we can characterize the performance of the closest and farthest user in a NOMA cluster, respectively.
\subsubsection{Mean Rate Coverage of a Cluster}
Although the individual rate coverage of a user in the representative cluster is a useful metric, it does not offer a complete insight related to the cumulative performance of the users in the representative NOMA cluster. As such, we consider analyzing the collective performance of all users  by defining the mean rate coverage of all users in the representative NOMA cluster as follows:
\begin{equation}
\mathcal{O}=\sum_{m=1}^{\bar{c}} \frac{\mathcal{C}^{(\cdot)}_{m}}{\bar{c}},
\end{equation}
where $(\cdot)= \mathrm{P}, \mathrm{I}$, and $\mathrm{worst}$ for perfect SIC, imperfect SIC, and worst case SIC, respectively.


\section{Relevant Distance Distributions for the Characterization of Interference}
In this section, we  characterize the relevant distributions of the distances  between the BS of a representative cluster and different intra-cell user devices that are ordered according to their distances. These distance distributions are crucial for deriving  the Laplace transforms of the  intra-cluster interferences  and the rate coverage analysis.

\begin{figure}[h]
\begin{center}
\includegraphics[width = 3.5in]{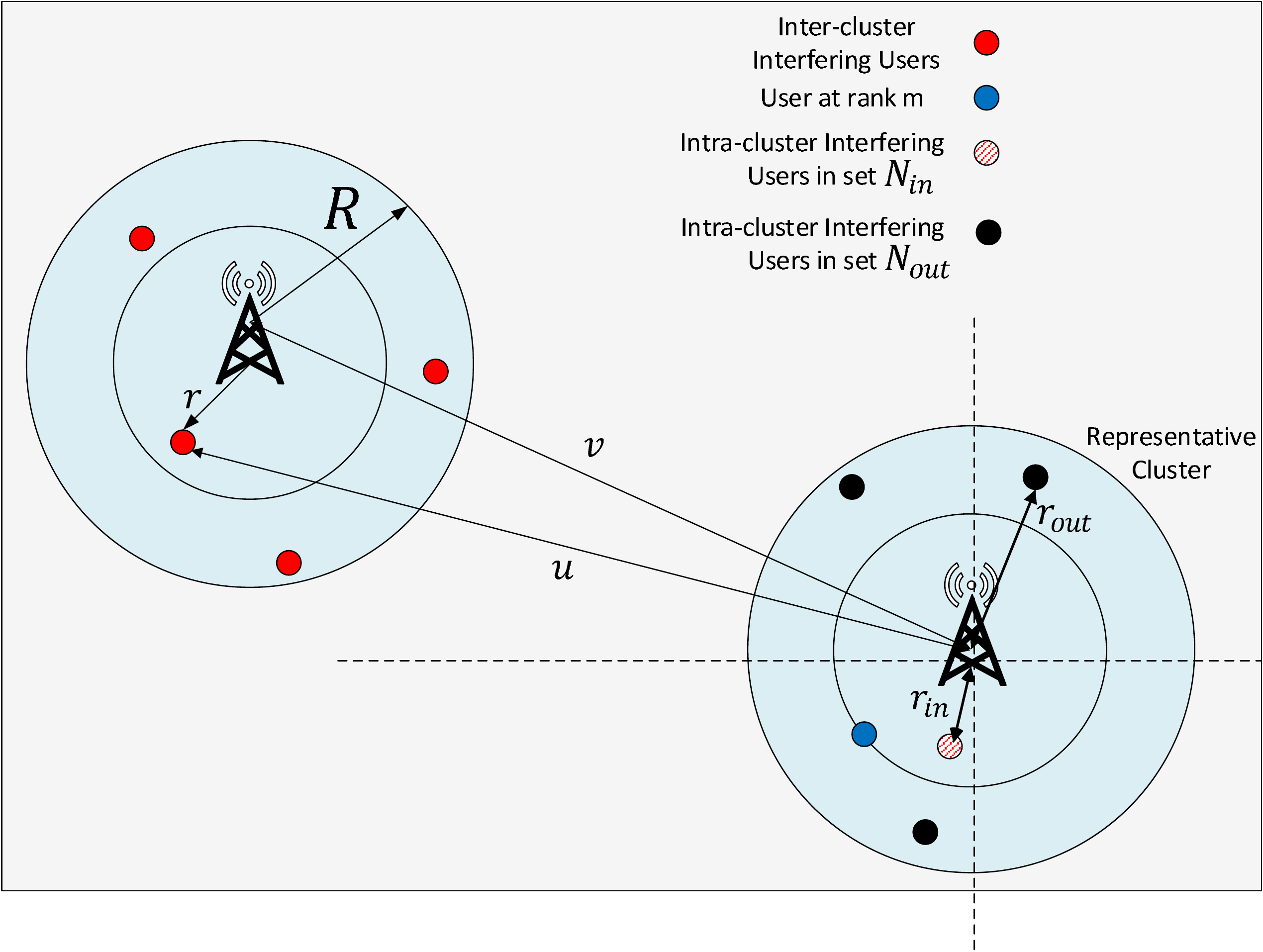}
\caption[c] { Illustration of the representative cluster, intra-cluster, and inter-cluster interferers along with their relevant distances.}
\label{faircomp2}
\end{center}
\end{figure}

As mentioned in the Approximation, we consider ranking the users in terms of their distances from the serving BS, which is located at the cluster center. 
Consequently, we characterize  the  distance distributions of the $m^{\mathrm{th}}$ closest node  as well as its respective intra-cluster interfering nodes from the BS. Note that the distance of any arbitrary intra-cluster device located at $y \in \mathbb{R}^2$ with respect to $x_0\in \mathbb{R}^2$ is i.i.d. and follows the uniform distribution given in \eqref{fy}. 
Subsequently, the distance $||y||=r$ of any arbitrary intra-cluster device from the representative BS follows a sampling distribution given by $f_r(r)$ in \eqref{fr}. Now we order $\bar{c}$ devices within the representative cluster with respect to the cluster center  such that $r_{(1)} \leq r_{(2)}\leq \cdots r_{(\bar{c})}$. The  distance distribution of the user  at rank  $m$ can thus be given as follows.

\begin{lemma}[Distribution of the Distance of the User at rank $m$] For a Matern cluster process, the distribution of the distance of a user at rank $m$ in the representative cluster, i.e., $r_{(m)}$ from its serving BS follows a Generalized Beta (GB) distribution of the first kind. The distance distribution can be derived as:
\begin{equation*}
f_{r_{(m)}}(r)= \mathrm{GB} (r,2,R,m,\bar{c}-m+1)=\frac{2 r^{2m-1} \left(1-\frac{r^2}{R^2}\right)^{\bar{c}-m}}{R^{2m} {B}(m,\bar{c}-m+1)}, 
\end{equation*}
where $B(\cdot)$ denotes the Euler Beta function.
GB distribution includes Beta distribution, Generalized Gamma distribution, and Pareto distribution as its special cases.
\end{lemma}

\begin{proof}
See \textbf{Appendix A}.
\end{proof}

Provided that the user of interest is at rank $m$, we now characterize the distribution of the distances of its corresponding intra-cluster interfering devices.
Note that the possible interfering nodes for a user at rank $m$ can lie at any place (depending on the perfect and imperfect SIC) except the location of the user at rank $m$. 
As such, the  set  of intra-cluster interferers can be partitioned
into two subsets, i.e., $\mathcal{N}^{x_0}_{\mathrm{in}} \in \{1, 2, \cdots, m-1\}$ and 
$\mathcal{N}^{x_0}_{\mathrm{out}} \in \{ m+1, \cdots, \bar{c}\}$ where
$\mathcal{N}^{x_0}_{\mathrm{in}}$ and 
$\mathcal{N}^{x_0}_{\mathrm{out}}$
represent  the  set  of  interfering  nodes  closer and  farther  to  the  reference  BS,  respectively,  compared  to  the  user at rank $m$. This set-up is demonstrated in Fig.~2. 
Consequently,  the distances from the representative BS to  devices in $\mathcal{N}^{x_0}_{\mathrm{in}}$ are i.i.d as shown by {\em Afshang et al.} in \cite{afshang1,afshang2,afshang3,afshang4}. The same property holds for the users in  $\mathcal{N}^{x_0}_{\mathrm{out}}$.

\begin{lemma}[Conditional Distribution of the Distances of Intra-Cluster Interfering Users]
Conditioned on the distance of the user of interest at rank $m$ (say $r_{(m)}=\hat{r}$),  the distribution of the distance $r_{\mathrm{in}}$ of any device in the set $\mathcal{N}^{x_0}_{\mathrm{in}} \in \{1, 2, \cdots, m-1\}$ from cluster center (BS) is i.i.d and can therefore be given by:
\begin{equation}\label{lemma2a}
f_{r_{\mathrm{in}}}(r_{\mathrm{in}}|\hat{r})=\frac{f_r(r_{\mathrm{in}})}{F_r(\hat{r})}=\frac{2 r_{\mathrm{in}}}{\hat{r}^2}, \quad r_{\mathrm{in}}<\hat{r}.
\end{equation}
Similarly,  the distribution of the distance of any device in the set $\mathcal{N}^{x_0}_{\mathrm{out}} \in \{m+1, \cdots, \bar{c}\}$ from cluster center (BS) is conditionally i.i.d. and can be given as:
\begin{equation}\label{lemma2b}
f_{r_{\mathrm{out}}}(r_{\mathrm{out}}|\hat{r})=\frac{f_r(r_{\mathrm{out}})}{1-F_r(\hat{r})}=\frac{2 r_{\mathrm{out}}}{R^2-\hat{r}^2}, \quad R\geq r_{\mathrm{out}}>\hat{r},
\end{equation}
where $r_{\mathrm{in}}$ and $r_{\mathrm{out}}$ are conditionally independent.
\end{lemma}
\begin{proof}
The proof can be done along the same lines as shown in \cite{afshang1}. For sake of completeness of the paper, we discuss it briefly in \textbf{Appendix B}.
\end{proof}
At this point, it is noteworthy that {\bf Lemma~2} utilizes the fact that the ordering of users in the set $\mathcal{N}^{x_0}_{\mathrm{out}}$ does not have any impact on the cumulative interference generated from all users of this set.
Therefore,   the interfering devices in the set $\mathcal{N}^{x_0}_{\mathrm{out}}$ can be sampled randomly without any specific ordering. Consequently, the  distance distributions of the  interfering users belonging to  set  $\mathcal{N}^{x_0}_{\mathrm{out}}$ are i.i.d. and can be given as in {\bf Lemma~2}. The same statements hold for the users in set $\mathcal{N}^{x_0}_{\mathrm{in}}$.

Nonetheless, for imperfect SIC scenarios, we need to include  specific interferences from the ranked users in set $\mathcal{N}^{x_0}_{\mathrm{in}}$ whose signals go undetected.  As such, for each combination $\mathbf{b}$, we need to include additional intra-cluster interferences from the specific users whose corresponding $b(i)=1$ in set $\mathcal{N}^{x_0}_{\mathrm{in}}$.  Given that the user of interest is at rank $m$, the distance distribution 
of a user at specific  rank $j<m$, can be derived by exploiting the conditional ranked distributions from the theory of order statistics as in the following.
\begin{lemma}[Conditional Distribution of the Distance of an Intra-Cluster Interfering User at Rank $j$]
Given that the user of interest is at rank $m$ (say $r_{(m)}=\hat{r}$),  the conditional distribution of the distance of $j$-th rank user, such that $j>m$ and $r_{(j)} > r_{(m)}=\hat r$, is the distribution of the $j-m$-th order statistics of $\bar{c}-m$ i.i.d. distance variables whose PDF  is the truncated PDF of $r$ as given in \eqref{lemma2b}. The conditional distance distribution of $j$-th rank user can thus be given as: 
\begin{equation}
f_{r_{(j)}|r_{(m)}=\hat{r}}(\gamma_j)=
\frac{2 \gamma_j (\bar{c}-m)! (\gamma_j^2-\hat{r}^2)^{j-m-1}
(R^2-\gamma_j^2)^{\bar{c}-j}
}
{\Gamma(j-m)(\bar c-j)! (R^2-\hat{r}^2)^{\bar{c}-m}}.
\end{equation}
Accordingly, the conditional distribution of the distance of $j$-th rank user, such that $j<m$ and $r_{(j)} < r_{(m)}=\hat r$, is the same as the distribution of the $j$-th order statistics of $m-1$ i.i.d. distance variables whose PDF is the truncated PDF of  $r$ as given in \eqref{lemma2a}. The conditional distance distribution of $j$-th rank user can thus be given as: 
\begin{equation}
f_{r_{(j)}|r_{(m)}=\hat{r}}(\gamma_j)=
\frac{2\Gamma(m) (\hat{r}^2-{\gamma_j^2})^{m-j-1}
(\gamma_j)^{2j-1}
}{\Gamma(m-j) \Gamma(j) \hat{r}^{2m-2}}.
\end{equation} 
\end{lemma}
\begin{proof}
The proof of the first part of the lemma follows by substituting the PDF from \eqref{lemma2b} and  its corresponding CDF into \eqref{A1} and replacing $\bar{c}$ with $\bar{c}-m$ and $m$ with $j-m$. Similarly, the  proof of the second part of the lemma follows by substituting the PDF from \eqref{lemma2a} and  its corresponding CDF into \eqref{A1} and replacing $\bar{c}$ with $m-1$ and $m$ with $j$.
\end{proof}

\section{Rate Coverage Analysis}
Using the distance distributions  derived in Section~IV, in this section, we  derive the Laplace transforms  of the intra-cluster interferences experienced by the transmission of  $m$-th rank user  considering both the perfect and imperfect SIC scenarios. We then derive the Laplace transform of the inter-cluster interference incurred at the representative BS by exploiting the distance distributions from geometric probability and a less complex upper bound of the Laplace Transform of the inter-cluster interference. Finally, we derive the rate coverage expressions for NOMA as well as OMA system.

\subsection{Laplace Transforms of the Intra-Cluster Interferences}

As defined in \eqref{perfectI}, it is evident that the intra-cluster interference with perfect SIC is due to all users that are located beyond the distance $\hat{r}=r_{(m)}$. Subsequently, the intra-cluster interference  defined in \eqref{perfectI} can  be rewritten as follows:
\begin{equation}\label{perfectII}
I^{m}_{\mathrm{intra}}= 
\sum_{y \in \mathcal{N}^{x_0}_{\mathrm{out}}}  P_u h_{y_{x_0}} ||y||^{-\alpha}.
\end{equation}
The Laplace transform of $I^{m}_{\mathrm{intra}}$ can then be given as follows.
\begin{lemma}[Laplace transform of the Intra-Cluster Interference with Perfect SIC] The intra-cluster interference experienced by the transmission of $m$-th ranked user in the cluster, with perfect SIC, can be given as follows:
\begin{equation}
\mathcal{L}_{I^{m}_{\mathrm{intra}}}(s)= 
\left(
\frac{2 s^{\frac{2}{\alpha}} \tilde{\mathbf{B}}[-\frac{R^\alpha}{s},-\frac{\hat{r}^\alpha}{s}, \frac{2+\alpha}{\alpha}, 0]
}{(R^2-\hat{r}^2)(-1)^{\frac{2}{\alpha}}\alpha}
\right)^{\bar{c}-m},
\end{equation} 
where $\tilde{\mathbf{B}}(\cdot)$ is the generalized incomplete Beta function.
\end{lemma}
\begin{proof}
See \textbf{Appendix~C}.
\end{proof}
For cases of practical interest such as $\alpha=4$, the result  presented in the Lemma~4 can be further simplified as follows.
\begin{corollary}[Laplace transform of the Intra-Cluster Interference with Perfect SIC, $\alpha=4$]
When $\alpha=4$, the Laplace transform of the intra-cluster interference experienced by the transmission of $m$-th ranked user in the cluster, with perfect SIC, can be further simplified as follows.
\begin{align*}
\mathcal{L}_{I^{m}_{\mathrm{intra}}}(s)=&\Scale[1]{
\left(
1 - \frac{\sqrt{P_u s}\left(\mathrm{tan}^{-1}\left(\frac{\hat{r}^{2}}{\sqrt{P_u s}}\right)-\mathrm{tan}^{-1}\left(\frac{R^{2}}{\sqrt{P_u s}}\right)\right)}{\hat{r}^2-R^2}\right)^{\bar{c}-m}},
\\
&\stackrel{(a)}{=}\left(1 - \frac{\sqrt{P_u s} \mathrm{tan}^{-1} \left( \frac{\sqrt{P_u s}(\hat{r}^2-R^2)}{P_u s +\hat{r}^2 R^2} \right)}{\hat{r}^2-R^2}\right)^{\bar{c}-m},
\end{align*}
where (a) is derived by using the property $\mathrm{tan}^{-1}(x)-\mathrm{tan}^{-1}(y)=\mathrm{tan}^{-1}\left(\frac{x-y}{1+xy}\right)$.
Note that, if $m=\bar{c}$, which represents the farthest user from the BS, there is no intra-cluster interference since $\bar c -m=0$. 
\end{corollary}
Further, an accurate 
second-order approximation of $\mathrm{tan}^{-1}(x)$ 
with  a  maximum  absolute  error  of
0.0053 rad  can be derived as 
$
\mathrm{tan}^{-1}(x) \approx \frac{\pi}{4}x+0.273x|1-x|
$.
With this approximation, the closed-form expression presented in Corollary~1 can be further simplified.

The intra-cluster interference incurred at a BS with imperfect SIC is defined in \eqref{imperfectI}. It can be seen that the intra-cluster interference is composed of two parts, i.e., the interference observed from all users in set
$\mathcal{N}^{x_0}_{\mathrm{out}} $
that are located beyond the distance $\hat{r}=r_{(m)}$ and the interference from users in set $\mathcal{N}^{x_0}_{\mathrm{in}}$  whose signals go undetected. 
The Laplace transform of the intra-cluster interference with imperfect SIC can then be derived in closed-form as follows.
\begin{lemma}[Laplace transform of the Intra-Cluster Interference with ImPerfect SIC] The intra-cluster interference experienced by the transmission of $m$-th ranked user in the cluster, with imperfect SIC, can be given as follows:
\begin{equation*}
\mathcal{L}_{I^{m,\mathbf{b}}_{\mathrm{intra}}}(s) \stackrel{(a)}{=} 
\mathcal{L}_{I^{m}_{\mathrm{intra}}}(s) \mathcal{L}_{I^{m}_{\mathrm{add}}}(s),
\end{equation*}
where (a) follows from the  fact that, conditioned on the serving distance of $m$-th user, $\mathcal{N}^{x_0}_{\mathrm{in}}$ and $\mathcal{N}^{x_0}_{\mathrm{out}}$ are independent and the interfering distance distributions of their respective users are also independent. 
Note that $\mathcal{L}_{I^{m}_{\mathrm{intra}}}(s)$ is given in Lemma~4
and $\mathcal{L}_{I^{m}_{\mathrm{add}}}(s)$ can be derived as follows:
\begin{align*}&
\mathcal{L}_{I^{m}_{\mathrm{add}}}(s){=} 
\prod_{j=1}^{m-1} 
\left(
\sum_{i=0}^{m-j-1}
K(i)
\frac{
s^{\frac{2+2i}{\alpha}} \mathbf{B}[\frac{-\hat{r}^\alpha}{s},1+\frac{2+2i}{\alpha},0]}{\hat{r}^{ 2 (j +  i)}}
\right)^{1-b(j)},
\end{align*}
where  $K(i)=\frac{2 \Gamma[m] (-1)^{i - \frac{2+2i}{\alpha} -1}}{\Gamma[j] \Gamma[i+1] \Gamma[m-j-i] \alpha}$ and $\mathbf{B}(\cdot)$ is the incomplete Beta function.
\end{lemma}
\begin{proof}
See \textbf{Appendix~D}.
\end{proof}

\subsection{Laplace transforms of the Inter-Cluster Interference}
Typically, the  Laplace transform of the inter-cluster interference  in PCP is characterized  for  a fixed-distance typical link (where the receiver is not chosen from the PCP)~\cite{martin}. Recently,  a more general approach is presented in \cite{afshang1,afshang2} to characterize the inter-cluster interference
in which both the transmitter and receiver can be chosen from the PCP. Also, it was shown in \cite{afshang1,afshang2} that both the intra/inter-cluster interfering distances can be modeled using Rice distributions for Modified Thomas Cluster Processes.

Given the definition of the inter-cluster interference in \eqref{inter}, $I_{\mathrm{inter}}=\sum_{x \in \Phi_m \backslash x_0} \sum_{y \in \mathcal{N}^{x}}  P_u h_{y_{x}} ||x+ y||^{-\alpha}$ and following a similar approach proposed in \cite{afshang1}, we can write the Laplace transform $\mathcal{L}_{\mathrm{inter}}(s)$ as follows:
\begin{align*}
&\mathcal{L}_{\mathrm{inter}}(s)\stackrel{(a)}{=}
\mathbb{E} \left[e^{-s \sum_{x \in \Phi_m \backslash x_0} \sum_{y \in \mathcal{N}^{x}}  P_u h_{y_{x}} ||x+ y||^{-\alpha} }\right],
\nonumber\\&\stackrel{(b)}{=}
\mathbb{E} \left[  \prod_{x \in \Phi_m \backslash x_0}   
\prod_{y \in \mathcal{N}^{x}} 
\mathbb{E}_{h}[ e^{-P_u h_{y_{x}} ||x+ y||^{-\alpha}}]
\right],
\nonumber\\&\stackrel{(c)}{=}
\mathbb{E} \left[  \prod_{x \in \Phi_m \backslash x_0}  
\left(
\mathbb{E}_y \left[\frac{1}{1+s P_u ||x+ y||^{-\alpha}}\right]
\right)^{\bar{c}}
\right],
\nonumber\\&\stackrel{(d)}{=}
\Scale[1]{
\mathrm{exp}\left[-\lambda_m  
\int_{\mathbb{R}^2}  
\left(1-
\left(\int_{\mathbb{R}^2}
\frac{1}{1+s P_u ||x+ y||^{-\alpha}} f_Y(y) dy\right)^{\bar{c}}
\right)
dx
\right]},
\nonumber\\&\stackrel{(e)}{=}
\Scale[1]{
\mathrm{exp}\left[-2 \pi \lambda_m  
\int_{0}^\infty   
\left(1-\left(
\int_{u} 
\frac{1}{1+s P_u u^{-\alpha}} f_U(u|v) du\right)^{\bar{c}}
\right)
 v dv
\right]},
\end{align*}
where (a) follows from the definition of the Laplace transform and (b) follows from the property of the exponential function. Conditioned on the distance  $v=||x||$ from  the representative BS to the cluster center $x \in \Phi_m$,  the distance of each user device within the cluster (whose center is located at $x$) to the representative BS is i.i.d. Subsequently,   (c) follows from the fact that the distances of all users in an interfering cluster to the reference BS are i.i.d., (d) follows from the PGFL of PPP since all cluster centers $x$ follow a homogeneous PPP, and (e) follows from the conversion of Cartesian to polar coordinates. The Laplace transform can thus be derived as follows.
\begin{lemma}[Laplace Transform of the Inter-Cluster Interference]
Given the distance $v$ as illustrated in Fig.~2, the conditional distribution of $u$ can be given for the cases (i) $v \geq R$ and (ii) $v < R$, respectively, as follows \cite{geom1,geom2}:
\begin{equation}\label{d1}
f_{U|{v\geq R}}(u)=\frac{u^2}{R^2}-\frac{2 u^2}{R^2}\mathrm{sin}^{-1}\left(\frac{v^2-R^2+u^2}{2 u v}\right),
\end{equation}
where $v-R\leq u\leq R+v$. For the case $v<R$, we have:
\begin{equation}\label{d2}
f_{U|{v < R}}(u)=
\begin{cases}
\frac{2 u}{ R^2}, & 0\leq u\leq R-v,\\
\frac{u^2}{R^2}-\frac{2 u^2}{R^2}\mathrm{sin}^{-1}\left(\frac{v^2-R^2+u^2}{2 u v}\right),& R-v\leq u\leq R+v.
\end{cases}
\end{equation}
Subsequently, $\mathcal{L}_{\mathrm{inter}}(s)$ can be defined using \eqref{d1} and \eqref{d2} as:
\begin{align}\label{d1d2}
&\mathcal{L}_{\mathrm{inter}}(s)=
\mathrm{exp}\left(-2 \pi \lambda_m  
\left(\mathcal{A}_1 +\mathcal{A}_2\right)
\right),
\end{align}
where
\begin{equation*}
{\mathcal{A}_1=
\int_{0}^{R}   
\left(1-\left(\underbrace{
\int_u 
\frac{1}{1+s P_u u^{-\alpha}} f_{U|v < R}(u) du}_{\mathrm{I}_1}\right)^{\bar{c}}
\right)
 v dv},
\end{equation*}
\begin{equation*}
\mathcal{A}_2=
\int_{R}^\infty   
\left(
1-\left(
\underbrace{
\int_u 
\frac{1}{1+s P_u u^{-\alpha}} f_{U|v\geq R}(u) du}_{\mathrm{I}_2}\right)^{\bar{c}}
\right)
 v dv.
\end{equation*}
\end{lemma}
The integrals $\mathrm{I}_1$ and $\mathrm{I}_2$ are solvable in closed-form as detailed in {\bf Appendix~E} and the Laplace transform as well as the rate coverage can be evaluated numerically using \texttt{MAPLE} and \texttt{MATHEMATICA}. 

Note that the Laplace transform in (d) can also be attained from the  PGFL of an MCP with fixed number of daughters per cluster  as defined below.
\begin{definition}[PGFL of the Matern Cluster Process]
The PGFL of the MCP given the number of nodes are fixed $\bar c$ per cluster
can be given as follows:
\begin{equation}
\mathbb{E}\left[\prod_{x \in \Phi_m}v(x) \right]=\mathrm{exp}\left[-\lambda \int_{\mathbb{R}^2} (1-{Z}(v(x))^{\bar c}) dx\right],
\end{equation}
where $\lambda$ denotes the intensity of the parent point process which is a homogeneous PPP in case of MCP and $Z(v(x))=\int_{\mathbb{R}^2} v(x+y) f(y) dy=\mathbb{E}_y[v(x+y)]$.
Since the original cluster process is stationary, the inter-cluster interference is independent of the position of the receiver~\cite{martin}.
\end{definition}
Now observing that $Z(v(x))^{\bar{c}}=\mathbb{E}_y[v(x+y)]^{\bar{c}}=\mathbb{E}_y[\frac{1}{1+s P_u ||x+ y||^{-\alpha}}]^{\bar{c}}$, we can apply Jensen inequality $\mathbb{E}_y [a]^{\bar{c}}\leq\mathbb{E}_y[a^{\bar{c}}]$ to derive an upper bound on the Laplace transform of the inter-cluster interference. That is 
\begin{align}\label{bound}
&\mathcal{L}_{\mathrm{inter}}(s)\stackrel{(a)}{\leq}
\int_{\mathbb{R}^2}\left(1-
\mathbb{E}_y\left[\frac{1}{(1+s P_u ||x+ y||^{-\alpha})^{\bar{c}}}\right]\right)dx,
\\
&=
\int_{\mathbb{R}^2} \int_{\mathbb{R}^2} 
\left(1-\frac{1}{(1+s P_u ||x+ y||^{-\alpha})^{\bar{c}}}\right) f_Y(y) dy dx,
\\
&\stackrel{(b)}{=}
\mathrm{exp}\left({-\pi \lambda_m (s P_u)^{\frac{2}{\alpha}} \bar{c} B[1 - \frac{2}{\alpha}, 
  \bar{c} + \frac{2}{\alpha}]}\right),
\end{align}
where (a) follows from applying the Jensen inequality and (b) follows from applying the cartesian-to-polar coordinate transformation and integration by parts~\cite{hasna}. 
\subsection{Rate Coverage Probability for NOMA} 
In this section, we derive the rate coverage probability  of a user at rank $m$ for all three cases, i.e., perfect SIC, imperfect SIC, and worst-case SIC. 

\subsubsection{Rate Coverage with Perfect SIC} 
Using the distance distribution of the user at rank $m$ derived in {\bf Lemma~1}
and the Laplace transforms of the inter- and intra-cluster interferences derived, respectively, in {\bf Lemma~4} and {\bf Lemma~6}, the $\mathrm{SINR}$ coverage of a user at rank $m$ can be derived as follows: 
\begin{align}\label{rate11}
&\mathcal{C}^{(\mathrm{P})}_m
=\mathbb{P}\left(\frac{P_u h_{y_{x_0}} \hat{r}^{-\alpha}}{I^{m}_{\mathrm{agg}}+N_0} \geq \gamma_m\right)
=\mathrm{exp}\left(-\frac{\gamma_m(I^{m}_{\mathrm{agg}}+N_0)}{P_u \hat{r}^{-\alpha}} \right),
\nonumber\\
&\stackrel{(a)}{=}
\int_0^{R} 
e^{-\frac{\gamma_m N_0}{P_u \hat{r}^{-\alpha}}}
\mathcal{L}_{I^{m}_{\mathrm{intra}}}
\left(
\frac{\gamma_m}{P_u \hat{r}^{-\alpha}}\right)
\mathcal{L}_{I_{\mathrm{inter}}}
\left(
\frac{\gamma_m}{P_u \hat{r}^{-\alpha}}\right)
f_{\hat r} (\hat r) d\hat{r},
\end{align}
where (a) follows from averaging over the distribution of $\hat{r}$ and the definition $I^{m}_{\mathrm{agg}}=I_{\mathrm{inter}}+I^m_{\mathrm{intra}}$. Since NOMA systems are typically interference limited, the expression in (a) can be further simplified by substituting $N_0=0$. Also, taking $\alpha=4$ and taking the upper bound of the Laplace transform of the inter-cluster interference in \eqref{bound}, we can simplify \eqref{rate11} as follows:
\begin{align*}
  \sum_{i=0}^{c-m}  
\int_0^1 
\frac{G(i) 
\left(1 - \sqrt{\gamma_m} z\mathrm{cot}^{-1}(\sqrt{\gamma_m} z)  \right)^i}
{e^{z R^2 \sqrt{\gamma_m}\bar{c} \lambda_m B[1/2, 1/2 + \bar{c}]}z^{i+1-\bar{c}}}
dz,
\end{align*}
where
\begin{equation*}
G(i)=\frac{  (\sqrt{\gamma_m} \mathrm{cot}^{-1}\sqrt{\gamma_m} - 1)^{\bar{c}-m-i}
  {\bar{c}-m\choose i}}{B[m, 1 + \bar c - m]}.
\end{equation*}

\subsubsection{Rate Coverage with Imperfect SIC} 
In this case, a given user at rank $m$ is prone to the interferences from all users 
in the set $\mathcal{N}^{x_0}_{\mathrm{out}}$ as well as from some users in the set $\mathcal{N}^{x_0}_{\mathrm{in}}$ whose signals go undetected. As such, using the distance distribution of the user at rank $m$ derived in {\bf Lemma~1}, the Laplace transforms of the intra- and inter-cluster interferences derived, respectively, in {\bf Lemma~5} and {\bf Lemma~6}, and the interference limited case, the $\mathrm{SIR}$ coverage of a user at rank $m$ in the representative cluster can be derived as follows:  
\begin{equation}\label{best1}
\mathcal{C}^{(\mathrm{I})}_m
=
\int_0^{R} 
\mathcal{L}_{I^{m, \mathbf{b}}_{\mathrm{intra}}}
\left(
\frac{\eta_m}{P_u \hat{r}^{-\alpha}}\right)
\mathcal{L}_{I_{\mathrm{inter}}}
\left(
\frac{\eta_m}{P \hat{r}^{-\alpha}}\right)
f_{\hat r} (\hat r) d\hat{r}.
\end{equation}

\subsubsection{Rate Coverage with Imperfect SIC-Worst Case} 
The rate coverage probability with the worst case mapping can be derived as $\mathcal{C}^{\mathrm{worst}}_{m}= {p}^{\mathrm{worst}}_{(m)}\mathcal{C}^{(\mathrm{P})}_m$, where $\mathcal{C}^{(\mathrm{P})}_m$ is given using \eqref{rate11}. Note that the derivation of the detection probability ${p}^{\mathrm{worst}}_{(m)}$ defined in \eqref{pworst} can be done
by replacing $\gamma_m$ with $\theta$ in \eqref{rate11} for each user at $j$-th rank.

\subsection{Rate Coverage Probability for OMA}
The rate coverage performance of a user at rank $m$ in a given OMA cluster can be derived as follows.
\begin{corollary}[Rate Coverage Probability of a User at Rank $m$ in OMA]
The rate coverage probability of a user at rank $m$ in TDMA-based OMA system can  be given by taking zero intra-cluster interference, using the Laplace transform of the inter-cluster interference from \eqref{d1d2} with $\bar{c}$=1, and  replacing $\gamma_m$ with $\gamma^{(\mathrm{oma})}_m$ as follows:
\begin{align}\label{rate1}
&\mathcal{C}^{(\mathrm{oma})}_m
\stackrel{(a)}{=}
\int_0^{R} 
e^{-\frac{\gamma^{(\mathrm{oma})}_m N_0}{P_u \hat{r}^{-\alpha}}}
\mathcal{L}^{(\mathrm{oma})}_{I_{\mathrm{inter}}}
\left(
\frac{\gamma^{(\mathrm{oma})}_m}{P_u \hat{r}^{-\alpha}}\right)
f_{\hat r} (\hat r) d\hat{r},
\\&
\stackrel{(b)}{=}
\gamma[\bar{c} + 1] \:_1\tilde{F}_1[m, 1 +\bar c,- K R^2],
\end{align}
where $K=(\gamma^{(\mathrm{oma})}_m)^{2/\alpha}\bar{c} \lambda_m B[1-2/\alpha, 2/\alpha + \bar{c}]$ and $:_1\tilde{F}_1[\cdot]$ is the regularized confluent Hypergeometric function. 
\end{corollary}
Note that (b) is derived by neglecting the noise, taking the bound of the inter-cluster interference  in \eqref{bound}, and solving the integral in (a) in  closed-form.

\section{Numerical Results}
In this section, we  investigate the performance of NOMA and OMA in clustered cellular networks. We conduct a comparative analysis with the conventional PPP-based cellular network model  to demonstrate that how two different modeling approaches may impact the accuracy of the conclusions related to the performance of NOMA versus OMA. The accuracy of the derived rate coverage expressions is validated by comparing them in the results with the Monte-Carlo simulations. The performance of the  uplink NOMA versus uplink OMA system is investigated as a function of the maximum coverage radius  of clusters, users per cluster, intensity of the BSs, and the target rate requirements per user.

The locations of users are drawn from a Poisson cluster process in a  square region with area $|A|=10 \times 10~\mathrm{km}^2$. The cluster centers are spatially distributed as a PPP with intensity $\lambda_m$ and the users are  scattered uniformly around them. We set the  threshold for  successful demodulation and decoding  as 0 dB.  The  number 
of  users per cluster is taken as  $\bar{c}=8$ and the radius of each cluster is set to $r=0.8$~km.
We set the path-loss exponent to
$\beta = 4$ and the thermal noise power density to $\sigma^{2}=1\times 10^{-14}$~W/Hz. The transmit powers of  users are taken as $P_u = 2$~W. The target rate requirement of each  user is taken as $R_{\mathrm{th}}=3$ bps/Hz. The values of the aforementioned  parameters  remain the same unless stated otherwise.

\subsection{Impact of the Coverage Radius of a BS}

\begin{figure}
\begin{minipage}[t]{0.5\textwidth}
\includegraphics[width = 3.55in]{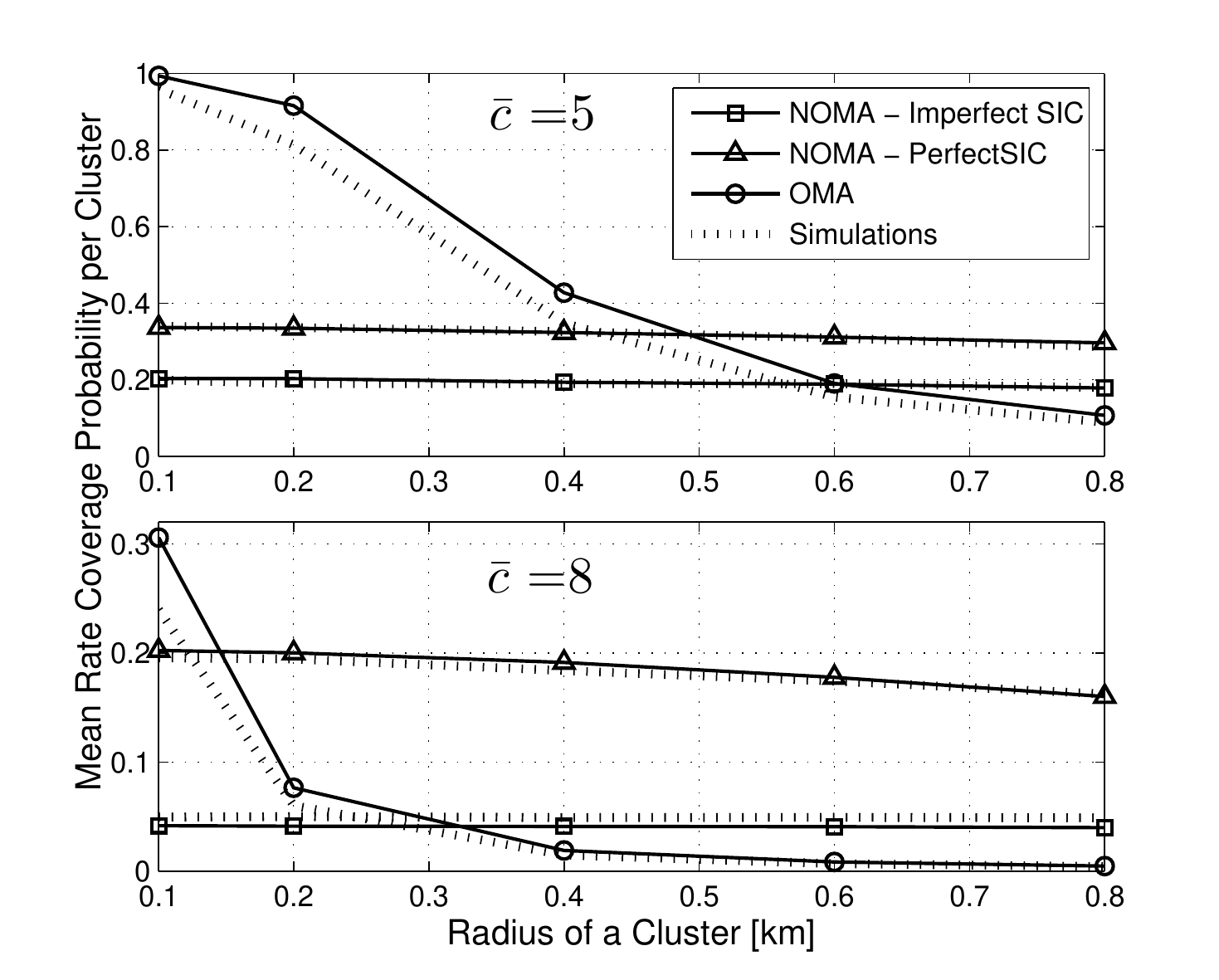}
\caption[c] { Mean rate coverage probability as a function of  the radius of each cluster in the single-tier cellular network, $\lambda_m|A|=2$.}
\label{nomaR}
\end{minipage}
\hfill
\begin{minipage}[t]{0.5\textwidth}
\includegraphics[width = 3.55in]{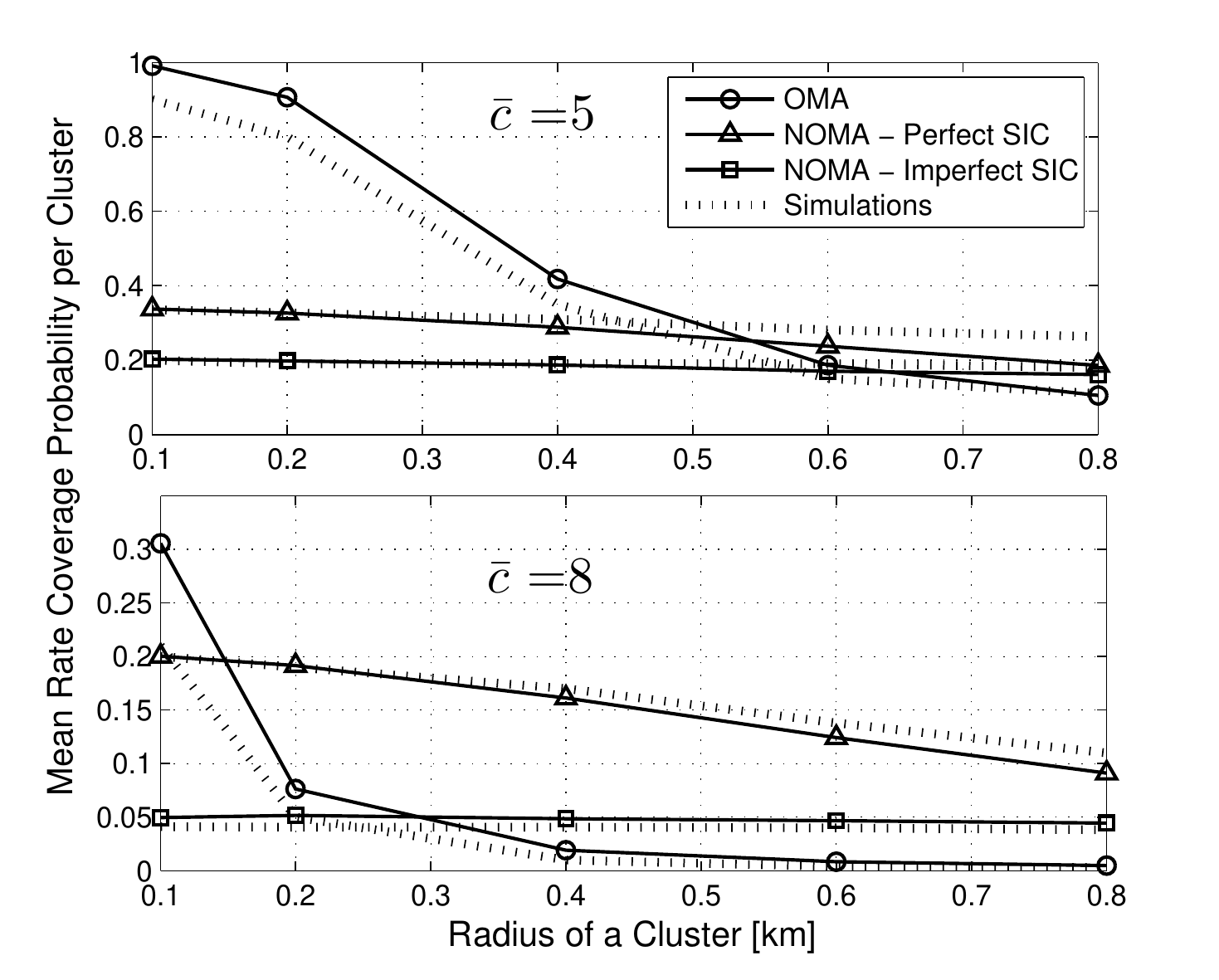}
\caption[c] { Mean rate coverage probability as a function of  the radius of each cluster in the single-tier cellular network, $\lambda_m|A|=8$.}
\label{nomaR1}
\end{minipage}
\end{figure}

\figref{nomaR} depicts the average rate coverage of $\bar c$ users in a NOMA representative cluster as a function of the maximum cell/cluster radius $R$ with $\lambda_m|A|=2$.  This scenario can be considered as intra-cell interference-limited due to low intensity of BSs. The results are compared with the mean rate coverage of $\bar{c}$ users in an equivalent OMA system where one user transmits at a time in each cell/cluster. 

First, it can be seen that the average rate coverage of the OMA cluster is higher than the NOMA cluster for small values of $R$. The reason is  no intra-cell interference in the OMA system and higher intra-cell interference in the NOMA system due to the vicinity of users and the representative BS. Nonetheless, as $R$ increases, the performance of OMA reduces significantly due to notable path-loss degradation of  weak users. However, interestingly, the performance of  NOMA system remains nearly the same since the effect of path-loss degradation gets nearly cancelled by the reduction of intra-cell interference.   As a result, beyond a certain coverage radius of a BS, the gains of NOMA become evident.

Further, we note that the performances of OMA and NOMA depend significantly on the number of users per cluster $\bar{c}$. As $\bar{c}$ increases, the mean coverage  probability reduces for both the OMA and NOMA systems. In OMA, the reduction is caused due to further splitting of resources, whereas in NOMA the reduction is caused due to the increased intra-cell interference. We also note that the rate coverage decay with $\bar c$ is more severe for OMA since the rate is a direct function of the amount of consumed resources. Due to this reason, NOMA starts to outperform OMA for relatively low values of $R$. {\em As such,  given a BS coverage radius $R$, it is important to select the correct  number of users in a   cluster to ensure channel distinctness.} 
Finally, it can be observed that the perfect SIC can improve the performance of NOMA significantly. Therefore, it is crucial to design efficient SIC strategies.

\subsection{Impact of the Intensity of BSs}

\figref{nomaR1} depicts the average rate coverage of $\bar c$ users in a NOMA representative cluster as a function of the maximum cell/cluster radius $R$ with $\lambda_m|A|=8$.   The results are compared with the mean rate coverage of $\bar{c}$ users in an equivalent OMA system where one user transmits at a time in each cell/cluster. 
The general conclusions and trends remain the same as in \figref{nomaR}. 
However, it can be seen that increasing the intensity of BSs reduces the coverage probability significantly. Further, in  low-inter-cell interference  scenarios (see \figref{nomaR}), we have observed that the performance of NOMA remains intact for increasing  values of $R$. On the contrary, at high intensity of BSs, \figref{nomaR1} shows that the NOMA performance degrades with increasing $R$. The reason is that the increasing  values of $R$ produce higher inter-cell interference as the users of neighboring cells are likely to be closer to the representative BS.
Note that the degradation of OMA is caused by both the path-loss degradation as well as inter-cell  interference, whereas in NOMA the degradation is mainly due to the increase of inter-cell interference. Finally, it can be seen that the rate coverage with imperfect SIC is least affected by the increasing intensity of BSs as the performance is limited by mainly the intra-cell interferences.

\begin{figure}[h]
\begin{center}
\includegraphics[width = 3.55in]{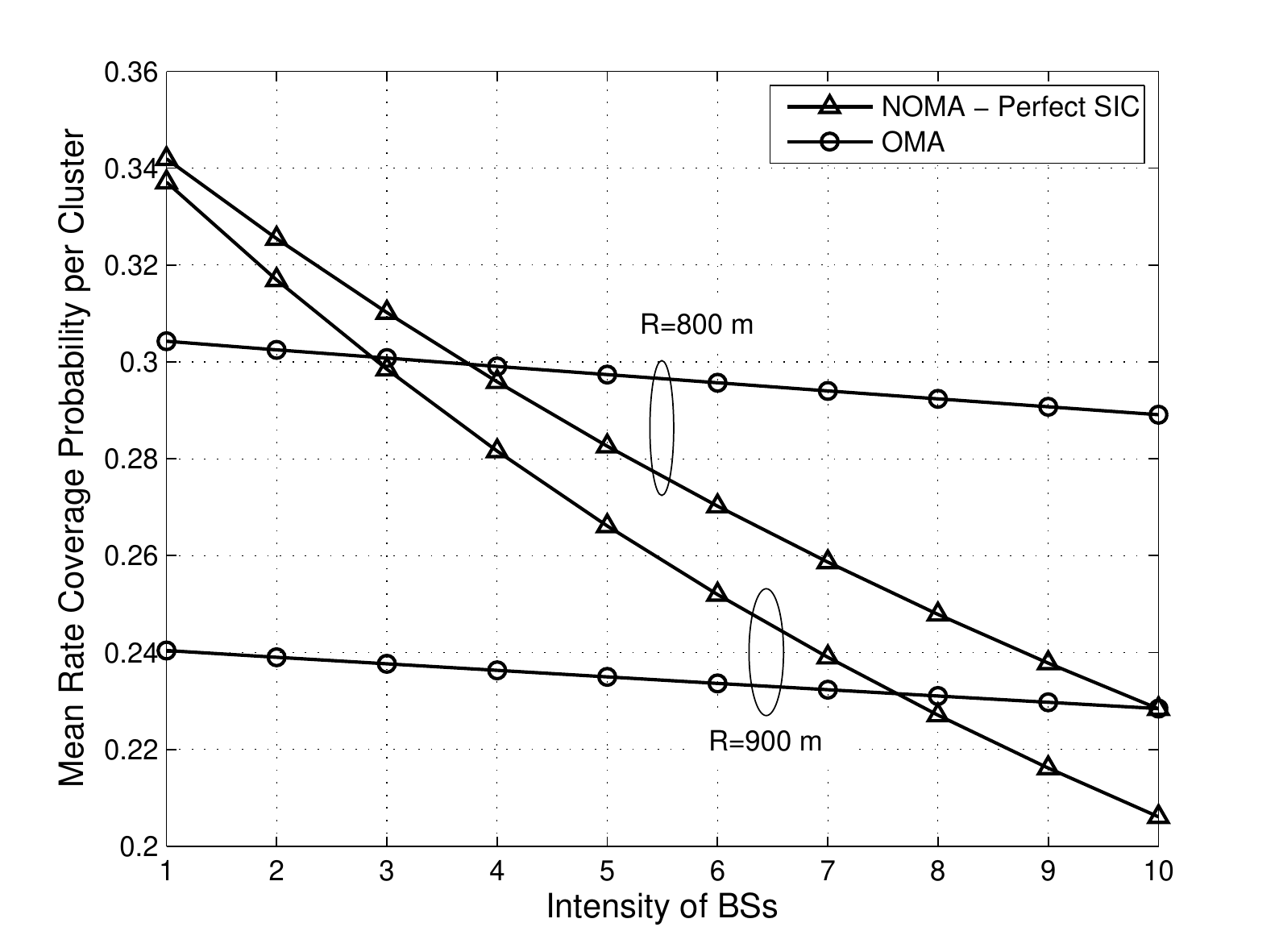}
\caption[c] { Mean rate coverage probability as a function of  $\lambda_m|A|$ in single-tier cellular network, $R_{\mathrm{th}}=1.5$~bps/Hz, $\bar{c}=8$.}
\label{Nm}
\end{center}
\end{figure}

\figref{Nm} demonstrates that increasing the intensity of BSs (or NOMA clusters) sharply degrades the performance of NOMA when compared to OMA. The reason is that  several users transmit in each NOMA cluster at the same time. On the other hand, there is only one user transmitting per BS in OMA; thus the performance of OMA is relatively less prone to the increasing intensity of BSs. 
On the other hand, the performance of OMA decays significantly with the increase in $R$ due to path-loss degradation, whereas in NOMA, the degradation is not significant as is also evident from \figref{nomaR} and \figref{nomaR1}. As such, {\em NOMA can potentially outperform OMA for a higher intensity of BSs if the coverage radius of BSs can be increased to ensure channel distinctness. }

\subsection{Analytical validation}
\figref{nomaR} and \figref{nomaR1} depict the mean rate coverage of $\bar c$ users in a cluster as a function of the maximum cell/cluster radius $R$ with $\lambda_m|A|=2$ and $\lambda_m|A|=8$, respectively. 
The impact of the distance approximation and the bound on the inter-cluster interference can be observed. The mismatch between the analysis and Monte-Carlo simulations is observed to increase with $\lambda_m$ which shows that the mismatch is greatly contributed by the bound on the Laplace transform of the inter-cluster interference. Also, the impact of the inter-cluster interference bound is more visible for OMA as the rate coverage expression of OMA depends mainly on the inter-cluster interference.

\subsection{Impact of the Desired User Rate Requirements}

\begin{figure}[h]
\begin{center}
\includegraphics[width = 3.55in]{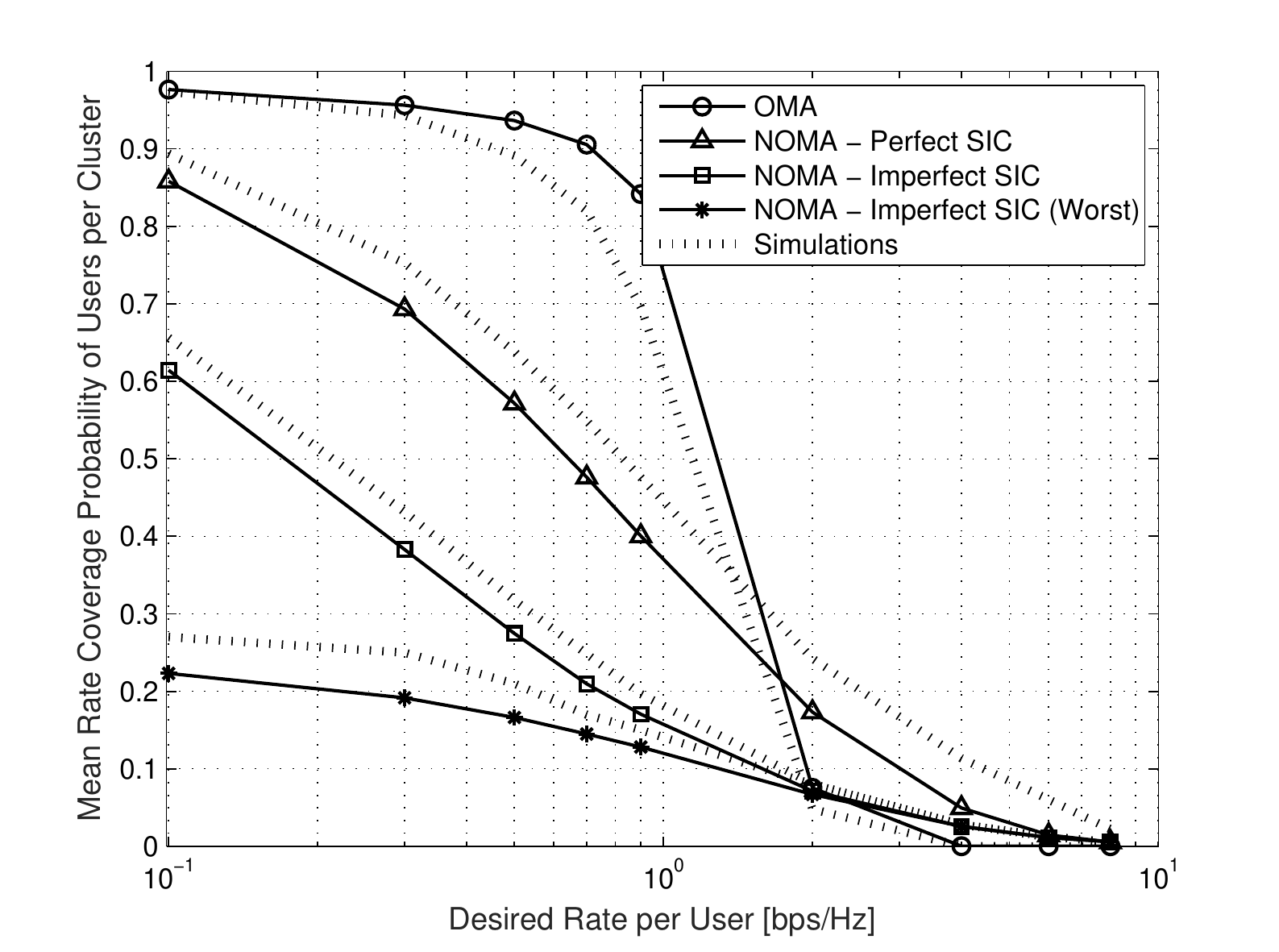}
\caption[c] { Mean rate coverage probability as a function of  the radius of each cluster in the single-tier cellular network, $\lambda_m|A|=8$.}
\label{rth1}
\end{center}
\end{figure}

\figref{rth1} depicts the mean rate coverage probability of all users in a NOMA cluster as a function of the users' rate requirements  $R_\mathrm{th}$ considering perfect SIC, imperfect SIC, and worst-case SIC. The results are compared with the mean rate coverage of an equivalent OMA system in which one user transmits at a time in each cell/cluster.
The performance of NOMA generally  turns out to be better than OMA for higher values of $R_\mathrm{th}$. The reason is that  OMA is more susceptible to higher values of $R_\mathrm{th}$ due to the  multiplicative factor ${R_\mathrm{th} \bar{c}}$ in the SINR threshold of OMA. Note that $\gamma_m^{(\mathrm{oma})}=2^{R_\mathrm{th} \bar{c}}-1$. Interestingly, it can be seen that the rate coverage of worst-case SIC bound is quite similar to the imperfect SIC at higher values of $R_{\mathrm{th}}$. However, the difference is visible at lower values of $R_{\mathrm{th}}$. The reason is that, at lower values of $R_{\mathrm{th}}$, the rate outages with imperfect SIC occur mainly due to the detection failures. As such, precisely mapping the impact of  detection failures becomes crucial.
Since the worst-case SIC bound assumes successful rate coverage of $m$-th user only when all relatively stronger users detected correctly, its worst performance is self-explanatory in this case.

\subsection{Impact of the Number of Users per Cluster}
\begin{figure}
\centering
\includegraphics[width = 3.55in]{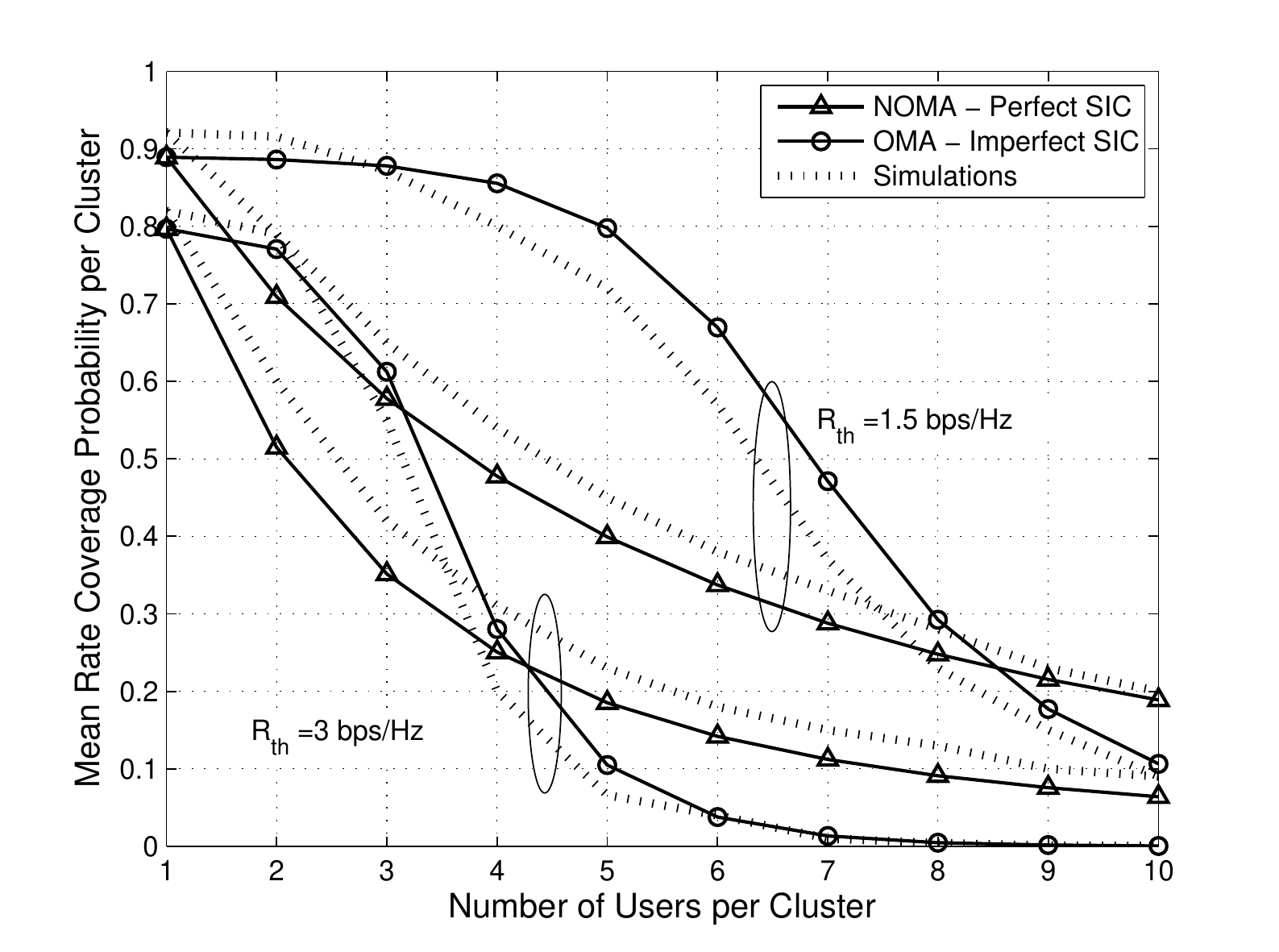}
\caption[c] {Mean rate coverage probability as a function of  the number of users per cluster in the single-tier cellular network, $\lambda_m|A|=8$, $R=0.8$~km.}
\label{nomarth1}
\end{figure}

\figref{nomarth1} represents the mean rate coverage probability of a cluster as a function of $\bar{c}$ considering $\lambda_m|A|=8$. It can be seen that the rate coverage of both OMA and NOMA generally decreases with  increasing  $\bar{c}$ and  target rate requirements $R_{\mathrm{th}}$ of the users. Clearly, with  increase in $\bar{c}$, the reduction is due to the increasing intra-cell and inter-cell interferences in NOMA and the reduced share of resources in OMA.  On the other hand, with  increasing $R_{\mathrm{th}}$, the  SINR threshold increases exponentially which reduces the coverage probability for both OMA and NOMA. However, the rate of decay of OMA is faster than NOMA for increasing  $R_{\mathrm{th}}$. The reason is that the SINR threshold of OMA has a  multiplicative factor ${R_\mathrm{th} \bar{c}}$ which is not the case in  NOMA.

\subsection{Poisson Point Process vs Poisson Cluster Process}

\begin{figure}[h]
\begin{center}
\includegraphics[width = 3.55in]{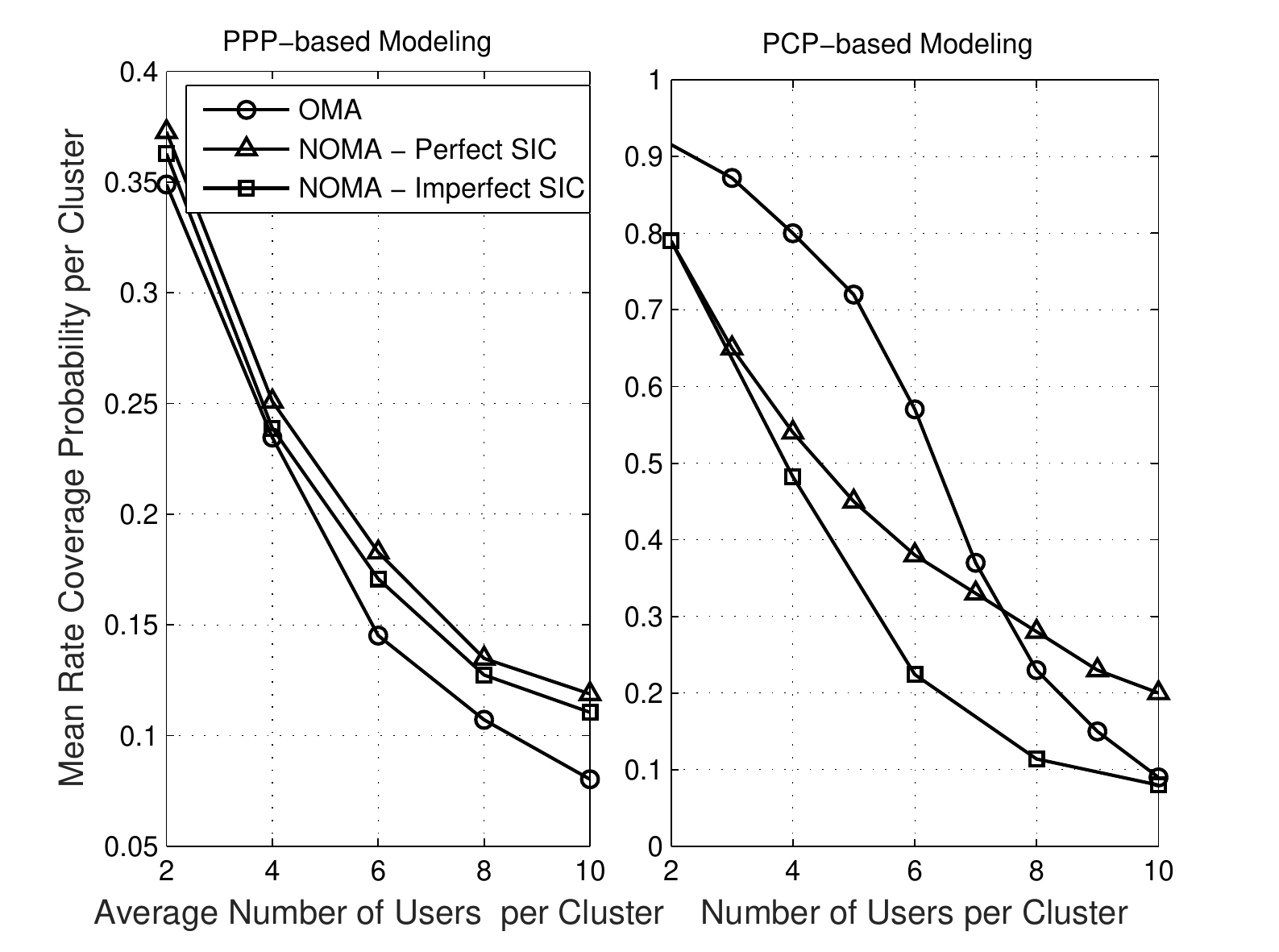}
\caption[c] { Mean rate coverage probability as a function of  the number of users per cluster for PCP-based modeling and total number of users as $\bar c*\lambda_m$ in the single-tier cellular network, $\lambda_m=8$, $R=800$~m, and $R_{\mathrm{th}}=1.5$~bps/Hz.}
\label{ppp}
\end{center}
\end{figure}

\figref{ppp} compares the performances of OMA and NOMA systems while comparing the two modeling approaches with equal intensity of BSs and users available in a square region. It is observed that PPP-based modeling shows reduced coverage probabilities for both OMA and NOMA systems, compared to PCP-based modeling. The reason is that there is no clustering of users around the BSs and they are spread independent of the locations of the BSs. As such, the impact of path-loss degradation is significant. Note that we assume minimum distance-based association  for PPP-based modeling. Also, the path-loss degradation is significant in PPP which benefits NOMA over OMA by lowering intra-cell interference among users. Consequently, NOMA with both perfect and imperfect SIC is observed to perform nearly the same and better than OMA. On the other hand, the PCP-based modeling shows that NOMA outperforms OMA for a certain range of users per cluster and that the perfect SIC outperforms imperfect SIC.

\section{Conclusion}

We characterize the rate coverage probability of a user in NOMA cluster who is at rank $m$ among all users  and the mean rate coverage probability of all users in the cluster considering perfect SIC, imperfect SIC, and imperfect worst case SIC.   In order to characterize the Laplace transforms of the intra-cluster interferences in closed-form considering both perfect and imperfect SIC scenarios, we derived novel distance distributions. The Laplace transform of the inter-cluster interference is then characterized by exploiting distance distributions from geometric probability as well as deriving an upper bound on the Laplace transform. Numerical results have been presented to validate the derived expressions. It has been shown that the average rate coverage of a NOMA cluster  outperforms its counterpart OMA cluster  in the cases of higher number of users per cell and higher target rate requirements. A comparison of PPP-based and PCP-based modeling has been conducted and it has been  shown that the PPP-based modeling provides optimistic results for the NOMA systems. Also, the distinctness of perfect and imperfect SIC is difficult to calibrate with PPP-based modeling of NOMA systems.

\section*{Appendix A\\Proof of Lemma~1}
\renewcommand{\theequation}{A.\arabic{equation}}
\setcounter{equation}{0}
Given the cluster members (users) are ranked according to their distances from the cluster center (BS), the distance distribution of a user at rank $m$ can be given using the standard theory of order statistics as follows~\cite{orderbook}:
\begin{equation}\label{A1}
f_{r_{(m)}}(r)=\frac{\bar{c}! (F_r(r))^{m-1} (1-F_r(r))^{\bar{c}-m} f_r(r)}{(\bar{c}-m)! (m-1)!}.
\end{equation}
For a Matern cluster process, cluster members are uniformly and identically distributed around their respective cluster center (BS); therefore, the unordered probability density function (PDF) and cumulative density function (CDF) of their distances from the cluster center is given by  \eqref{fr} and $r^2/R^2$, respectively. As such, \eqref{A1} can be rewritten as:
\begin{equation}\label{A2}
f_{r_{(m)}}(r)=\frac{\bar{c}! (\frac{r^2}{R^2})^{m-1} (1-\frac{r^2}{R^2})^{\bar{c}-m} (\frac{2 r}{R^2})}{(\bar{c}-m)! (m-1)!}.
\end{equation}
Using the definition of Gamma function $\Gamma(x)=(x-1)!$, the definition of the Euler Beta function ${B}(p,q)=\frac{\Gamma(p)\Gamma(q)}{\Gamma(p+q)}$, and doing some algebraic manipulations, we can write \eqref{A2} as:
\begin{equation}\label{A3}
f_{r_{(m)}}(r)=\frac{2\frac{r^{2m-1}}{R^{2m}}(1-\frac{r^2}{R^2})^{\bar{c}-m} }{{B}(m,\bar{c}-m+1)}.
\end{equation}
Now using the definition of the Generalized Beta (GB) distribution for a random variable $Z$, i.e.,
\begin{equation}
f_Z(z)=\frac{|a| z^{ap-1} (1-(1-c)(z/b)^a)^{q-1}}{b^{a p} \mathcal{B}(p,q) (1+c(z/b)^a)^{p+q}}, \quad 0 \leq z^a \leq \frac{b^a}{1-c},
\nonumber
\end{equation}
and substituting $c=0$, $z=r$, $a=2$,  $b=R$, $p=m$, and $q=\bar{c}-m+1$, the result in {\bf Lemma~1} can be readily proved. Note that the GB distribution with $c=0$ becomes a GB distribution of the first kind.

\section*{Appendix B\\Proof of Lemma~2}
\renewcommand{\theequation}{B.\arabic{equation}}
\setcounter{equation}{0}
The  joint distribution of the distances of the users in $\mathcal{N}^{x_0}_{\mathrm{out}} \in \{m+1, m+2, \cdots, \bar{c}\}$, i.e., 
$f_{r_{(m+1)}, r_{(m+2)}, \cdots, r_{(\bar{c})}}$ 
can be derived along the lines of~\cite{afshang1} as follows:
\begin{align}\label{app1} 
&\stackrel{(a)}{=}
\bar{c}! \prod_{i=m}^{\bar c} {f_{r}(r_{(i)})}
\int \cdots\int 
\prod_{i=1}^{m-1} {f_{r}(r_{(i)})} d r_{(1)}\cdots dr_{(m-1)},
\nonumber\\&
\stackrel{(b)}{=}
\frac{\bar{c}!}{(m-1)!} \prod_{i=m}^{\bar c} f_{r}(r_{(i)})
\left(\int_0^{r_{(m)}} f_{r}(r) d r\right)^{m-1},
\nonumber\\&
{=}
\frac{\bar{c}!}{(m-1)!} \prod_{i=m}^{\bar c} {f_{r}(r_{(i)})}
(F_{r}(r_{(m)}))^{m-1},
\end{align}
where (a) follows from averaging out the distance variables $r_{(1)}, r_{(2)}, \cdots, r_{(m-1)}$  and (b) follows from the i.i.d. property.
Now given the distance $r_{(m)}=\hat r$, the
conditional joint distribution of  the distances of the users in $\mathcal{N}^{x_0}_{\mathrm{out}} \in \{m+1, m+2, \cdots, \bar{c}\}$, i.e., 
$f_{r_{(m+1)}, r_{(m+2)}, \cdots, r_{(\bar{c})}|r_{(m)}} (r_{(m+1)}, r_{(m+2)}, \cdots, r_{(\bar{c})}|r_{(m)})$ can be given as follows:
\begin{equation}\label{app2}
 (\bar{c}-m)! 
\prod_{i=m+1}^{\bar c} \frac{f_{r}(r_{(i)})}{1-F_{r}(r_{(m)})}, \quad r_{(m)} \leq r_{(i)}.
\end{equation}
It is noteworthy that, for perfect SIC scenario, the interference experienced at the BS from users in set $\mathcal{N}^{x_0}_{\mathrm{out}}$ is cumulative therefore  the interfering users can be chosen without any specific ordering. Subsequently, the permutation $(\bar{c}-m)!$  in \eqref{app2} will not appear and the conditional joint distribution of the unordered set of distance random variables is given as:
\begin{equation}\label{app3}
\Scale[1]{f_{r_{m+1}, \cdots, r_{\bar{c}}|\hat r} (r_{m+1}, \cdots, r_{\bar{c}}|\hat r)
=\prod_{i=m+1}^{\bar{c}} \frac{f_{r}(r_{i})}{1-F_{r}(r_{(m)})} },
\end{equation}
where $r_{(m)} \leq r_{i}$. Consequently,  the product of truncated distributions  in \eqref{app3}   implies that the random variables of the unordered set are i.i.d. Therefore, the distribution of $m-1$ i.i.d. distance variables can be given as in {\bf Lemma~2}. Using similar arguments, the conditional PDF of $r_{\mathrm{in}}$ can be derived.

\section*{Appendix C\\Proof of Lemma~4}
\renewcommand{\theequation}{C.\arabic{equation}}

The Laplace transform of the intra-cluster interference with perfect SIC  as defined in \eqref{perfectII}, can  be derived as follows:
\begin{align*}
&\mathcal{L}_{I^{m}_{\mathrm{intra}}}(s)=\mathbb{E}\left[\mathrm{exp}\left(-s \sum_{y \in \mathcal{N}^{x_0}_{\mathrm{out}}}  P_u h_{{y}_{x_0}} ||y||^{-\alpha}\right)\right],
\\
&{=}\mathbb{E}_{\mathcal{N}^{x_0}_{\mathrm{out}}}\left[\prod_{y \in \mathcal{N}^{x_0}_{\mathrm{out}}}  \mathbb{E}_{h_{y_{x_0}}} \left[\mathrm{exp}\left(- s P_u h_{{y}_{x_0}} ||y||^{-\alpha}\right)\right]\right],
\\
&\stackrel{(a)}{=}\mathbb{E}_{\mathcal{N}^{x_0}_{\mathrm{out}}}\left[\prod_{y \in \mathcal{N}^{x_0}_{\mathrm{out}}}  \frac{1}{1+ s P_u||y||^{-\alpha}}\right],
\\
&\stackrel{(b)}{=}
\left(
\int_{\hat r}^{R} 
\frac{1}{1+ s P_u r_{\mathrm{out}}^{-\alpha}} f_{r_{\mathrm{out}}} (r_{\mathrm{out}}|\hat r) d r_{\mathrm{out}}\right)^{\bar c-m}.
\end{align*}
Note that (a) follows from the  definition of the Laplace transform of $h_{y_{x_0}}$ which is exponentially distributed unit mean random variable and (b) follows from the conversion of Cartesian to polar coordinates as well as the fact that the distances from the interfering devices in set $\mathcal{N}^{x_0}_{\mathrm{out}}$ to the BS,
conditioned on $\hat r$, are i.i.d. random variables. The integral in (b) can be solved in closed-form as follows:
\begin{equation}\label{c}
\Scale[1.2]{\frac{R^{2+\alpha}
\:_2F_1[1,1+\frac{2}{\alpha}, 2+\frac{2}{\alpha},-\frac{R^\alpha}{s P_u}]
-\hat{r}^{2+\alpha}
\:_2F_1[1,1+\frac{2}{\alpha}, 2+\frac{2}{\alpha},-\frac{\hat{r}^\alpha}{s P_u}]
}{s (P_u/2) (\alpha+2) (R^2-\hat{r}^2)}}.
\end{equation}
Finally, {\bf Lemma~3} can be readily obtained after simplifying the Gauss hyper-geometric functions in \eqref{c} as $\:_2F_1[a,b,b+1,z]=b z^{-b} \mathbf{B}_z(b,1-a)$ \cite[Eq. 07.23.03.0084.01]{mathematica} and then using the definition of the Generalized incomplete Beta function as $\tilde{\mathbf{B}}_{z_1,z_2}(x,y)=\mathbf{B}_{z_2}(x,y)-\mathbf{B}_{z_1}(x,y)$.

\section*{Appendix D\\Proof of Lemma~5}
\renewcommand{\theequation}{D.\arabic{equation}}
The Laplace transform of the intra-cluster interference experienced by the transmission of user at rank $m$, as defined in \eqref{imperfectI}, can  be derived as follows:
\begin{align}
&
\mathcal{L}_{I^{m}_{\mathrm{add}}}(s)
=
\mathbb{E}\left[\mathrm{exp}\left(- s \sum_{y_{(j)} \in \mathcal{N}^{x_0}_{\mathrm{in}}} (1-b(j)) P_u h_{{y}_{x_0}} ||y_{(j)}||^{-\alpha}\right)
\right],
\nonumber\\&
\stackrel{(a)}{=}
\mathbb{E}_{\mathcal{N}^{x_0}_{\mathrm{in}}}
{
\left[\prod_{y_{(j)} \in \mathcal{N}^{x_0}_{\mathrm{in}}}  
\mathbb{E}_{h_{{y}_{x_0}}}[e^{-s (1-b(j)) P_u h_{{y}_{x_0}} ||y_{(j)}||^{-\alpha}}]
\right]},
\nonumber\\
&\stackrel{(b)}{=}
\mathbb{E}_{\mathcal{N}^{x_0}_{\mathrm{in}}}
\left[\prod_{y_{(j)} \in \mathcal{N}^{x_0}_{\mathrm{in}}}  
\left(\frac{1}{1+ s (1-b(j)) P_u||y_{(j)}||^{-\alpha}} \right)^{1-b(j)}\right],
\nonumber\\
&\stackrel{(c)}{=}
\prod_{j=1}^{m-1} 
\left(
\int_{0}^{\hat r} 
\frac{1}{1+ s P_u r_{(j)}^{-\alpha}} f_{r_{(j)}} (r_{(j)}|\hat r) d r_{(j)}
\right)^{1-b(j)},
\end{align}
where (a) follows by applying the properties of the exponential function,
 (b) follows from  the Laplace transform of $h_{y_{x_0}}$ which is exponentially distributed unit mean random variable and  by shifting the term $1-b(j)$ into the powers. Note that $b(j) \in \{0,1\}$ and if $b(j)=1$ the interference from $i$-th ranked user is absent as its signal gets detected and if $b(j)=0$  the interference from $j$-th ranked user is present since its signal goes undetected, Finally, (c) follows from the conversion to polar coordinates and exploiting the following assumption.

{\bf Assumption:} Note that the distributions of ranked user devices in set $\mathcal{N}^{x_0}_{\mathrm{in}}$ are correlated.  The dependence is however weak since they are multiplied with random fading channel gains~\cite{tabassum2015spectral}. Thus, we ignore this dependence and subsequently interchange the operation of products and integration. 

Finally, substituting the distributions from Lemma~3, we apply the Binomial expansion and solve the integral using\cite[Eq. 07.23.03.0084.01]{mathematica} to obtain {\bf Lemma~5}.

\section*{Appendix E\\Proof of Lemma~6}
By applying the McLaurin Series expansion $\mathrm{sin}^{-1}(z)=\sum_{k=0}^\infty \frac{\Gamma[k+0.5]}{\sqrt{\pi} (2k+1) k!} z^{2k+1}$ and the Binomial expansion, $\mathrm{I}_2$ can be solved in closed-form as follows:
\begin{align}\label{aa}
\sum_{k=0}^{\infty} \sum_{j=0}^{2k+1}&
\frac{C\tilde{\mathbf{B}}(\frac{-(v-R)^\alpha}{s},\frac{-(v+R)^\alpha}{s},1+\frac{4-2j+2k}{\alpha},0)}{\alpha R^2 (-1/s)^{\frac{4-2j+2k}{\alpha}}} 
\nonumber\\&-\frac{\tilde{\mathbf{B}}(\frac{-(v-R)^\alpha}{s},\frac{-(v+R)^\alpha}{s},1+\frac{3}{\alpha},0)}{\alpha R^2 s^{-\frac{3}{\alpha}} (-1/s)^{\frac{3}{\alpha}}},
\end{align}
where 
$
C=\frac{\Gamma[k+0.5] {2k+1 \choose j} }{\sqrt{\pi} (2k+1) k! (v^2-R^2)^{-j} v^{2k+1} 2^{2k}}.
$ Similarly, the first part of $\mathrm{I}_1$ can be given in closed-form as $\frac{2 \mathbf{B}(-\frac{(R-v)^\alpha}{s}, 1 + \frac{2}{\alpha}, 0)}{ (-1/s)^{1+\frac{2}{\alpha}}\alpha (R^2) s }$ and second part of $\mathrm{I}_1$ can be solved in closed-form by replacing $v-R$ with $R-v$ in \eqref{aa}. Interested readers can go through \cite{geom2} for further details.

\bibliography{IEEEfull,Ref1}
\bibliographystyle{IEEEtran}

\end{document}